\documentclass{article}
\usepackage[utf8]{inputenc}
\usepackage[margin=1in]{geometry}
\usepackage{setspace}
\usepackage{amssymb}
\usepackage{amsmath}
\usepackage{amsthm}
\usepackage{bbm}
\usepackage{graphicx}
\usepackage{mwe}
\usepackage{color,soul}
\graphicspath{ {../images/} }
\usepackage{multirow}

\newtheorem{lemma}{Lemma}

\usepackage[algoruled,boxed,lined]{algorithm2e}
\usepackage{natbib}

\usepackage{thmtools}
\usepackage{thm-restate}
\newcommand{\pconv}{\xrightarrow{p}}
\newcommand{\ignore}[1]{}
\DeclareMathOperator{\plim}{plim}
\usepackage[algoruled,boxed,lined]{algorithm2e}

\providecommand{\keywords}[1]
{
  \small	
  \textbf{\textit{Keywords---}} #1
}

\usepackage{hyperref}
\hypersetup{
    colorlinks=true,
    linkcolor=blue,
    filecolor=magenta,      
    urlcolor=cyan,
    allcolors=blue,
}

\title{Achieving Reliable Causal Inference with Data-Mined Variables: A Random Forest Approach to the Measurement Error Problem}
\author{Mochen Yang$^{1}$, Edward McFowland III$^{1}$, Gordon Burtch$^{1}$, Gediminas Adomavicius$^{1}$ \\
\small $^{1}$University of Minnesota, Carlson School of Management}
\begin{document}
%%%%%%%%%%%%%%%%

\maketitle

\begin{abstract}
    Combining machine learning with econometric analysis is becoming increasingly prevalent in both research and practice. A common empirical strategy involves the application of predictive modeling techniques to ``mine" variables of interest from available data, followed by the inclusion of those variables into an econometric framework, with the objective of estimating causal effects. Recent work highlights that, because the predictions from machine learning models are inevitably imperfect, econometric analyses based on the predicted variables are likely to suffer from bias due to measurement error. We propose a novel approach to mitigate these biases, leveraging the ensemble learning technique known as the random forest. We propose employing random forest not just for prediction, but also for generating instrumental variables to address the measurement error embedded in the prediction. The random forest algorithm performs best when comprised of a set of trees that are individually accurate in their predictions, yet which also make ``different" mistakes, i.e., have weakly correlated prediction errors. A key observation is that these properties are closely related to the relevance and exclusion requirements of valid instrumental variables. We design a data-driven procedure to select tuples of individual trees from a random forest, in which one tree serves as the endogenous covariate and the other trees serve as its instruments. Simulation experiments demonstrate the efficacy of the proposed approach in mitigating estimation biases, and its superior performance over three alternative methods for bias correction.
\end{abstract}


\doublespacing
%%%%%%%%%%%%%%%%%%%%%%%%%%%%%%%%%%%%%%%%%%%%%%%%%%%%%%%%%%%%%%%%%%%%%%

\section{Introduction}

Advances in predictive machine learning have enabled researchers to extract useful information from various types of data, such as text and images, which would otherwise be difficult or costly to codify at scale. For example, recent academic work has highlighted that cutting-edge prediction techniques are now capable of inferring the socioeconomic properties of a localized population (e.g., income/racial distribution) from the models and makes of cars appearing in Google Street View images \citep{gebru2017using} and detecting adverse drug events based on drug attributes \citep{ryu2018deep}. These measurements, now available at scale and with little cost, can \textit{enable} empirical investigations of important questions in economics, health care, and many other domains.

Indeed, many researchers have begun doing exactly that, first using predictive machine learning to construct or populate a variable of interest, e.g., using text mining tools to predict text sentiment, and then including that variable into econometric models as an independent covariate. This practice has become prevalent in multiple social science domains, including economics \citep{jelveh2015political}, political science \citep{Fong2017}, and management \citep{yang2018mind}.

However, recent studies have also noted that attempts to draw inferences based on this recipe are likely to suffer from endogeneity due to measurement error \citep{yang2018mind}. This is because predictions from machine learning models are inevitably imperfect, and prediction errors carry over to subsequent econometric models as measurement error, leading to biased and inconsistent parameter estimates. Measurement error may lead to overestimation or underestimation of coefficients \citep{loken2017measurement}, and the degree of bias can be substantial, even when the machine learning model achieves reasonable predictive performance \citep{yang2018mind}. The estimation biases stemming from measurement error in machine-learning-generated covariates can thus undermine the validity of subsequent causal inferences and decision making. 

In this paper, we propose a novel approach to address this problem. Our approach is based on the notion of instrumental variable regression, a well-established method of resolving endogeneity in the econometrics literature, including endogeneity deriving from measurement error \citep{greene2003econometric}. We make use of a notably unique feature of this problem setting of applying machine learning followed by regression. Specifically, we leverage the fact that predictive machine learning models are typically trained and evaluated using data for which true labels (assumed to be perfectly measured) are available and which are used to quantify prediction error and model performance. This perfectly measured set of data offers a \textit{unique opportunity} to overcome difficulties commonly associated with evaluating the validity of instruments.

To find candidate instruments, we rely on \textit{random forest} \citep{Breiman2001}, an ensemble learning approach that aggregates a collection of individual decision trees (weak learners) to arrive at accurate predictions. Prior work has demonstrated that the performance of random forests depends jointly on (i) the degree to which trees comprising the forest yield correlated predictions and (ii) the degree to which the the trees yield weakly correlated prediction errors \citep{Breiman2001,bernard2010study}. We demonstrate that these notions are closely related to the relevance and exclusion criteria that underpin valid instrumental variables. Based on this, we look to explore the random forest ensemble to identify sets of individual trees such that one tree's predictions might serve as the endogenous covariate in the econometric model of interest, while other trees' predictions serve as its instruments, thereby alleviating estimation biases due to measurement error. Drawing on theories from the instrumental variable and random forest literature, we develop algorithms to implement this idea, empirically selecting the best set of individual trees to mitigate biases in coefficient estimates. We term our procedure \textit{ForestIV}.

We conduct two sets of comprehensive simulation experiments, considering the cases where a data-mined covariate is \textit{continuous} vs. \textit{binary}, and thus respectively suffers from either continuous measurement error or misclassification. In both cases, we show that ForestIV can effectively mitigate estimation biases. We also report sensitivity analyses of ForestIV and benchmark its performance against three alternative bias correction methods.

It should be noted that ForestIV provides a \textit{general-purpose} method for correcting bias with machine-learning-generated covariates, whether derived from structured or unstructured data (e.g., text or images). For scenarios involving structured data, random forest is widely applicable to various supervised machine learning problems, and, for a large number of real world prediction problems, random forest is among the most accurate techniques \citep{fernandez2014we}. However, even in scenarios involving unstructured data, where other techniques (e.g., deep neural networks) may be the state-of-the-art, random forest can be usefully combined with them. For example, random forest can be stacked with the intermediate representations learned by a neural network; that is, the output of an intermediate layer of the network, which encodes informative, high-level features learned from the unstructured data, can serve as the input features to the random forest algorithm. Notably, this practice is quite common in transfer learning, where a supervised machine learning model is built based on features produced by another technique \citep{goodfellow2016deep}.

Our paper makes several notable contributions. First, we theoretically and empirically show that the proposed ForestIV approach effectively addresses the estimation biases in econometric models due to measurement error in machine-learning-generated covariates. Therefore, ForestIV improves the robustness of causal inferences and decisions derived from procedures combining machine learning with econometric analyses. Second, we design data-driven procedures that leverage the labeled data (used to build and evaluate the machine learning model) to empirically select instruments that are most suitable for bias correction purposes. Third, ForestIV represents a novel method to automatically obtain candidate instruments from the output of the random forest technique. This provides a viable solution to the often-challenging problem of identifying valid instruments. 

\section{ForestIV with Continuous Endogenous Covariate}

Our work is informed by the econometrics literature on measurement error, instrumental variables, and generated regressors, as well as the machine learning literature on random forests. We provide a review of relevant literature in Appendix \ref{Literature}. In this section, we provide a description of the measurement error problem resulting from machine learning predictions of a \textit{continuous} covariate. We then introduce the theoretical justifications and implementation details of our proposed ForestIV solution.

\subsection{Continuous Measurement Error Problem Formulation}

We begin by setting up the measurement error problem for a continuous covariate. Note that, for expositional simplicity, we use a simple linear regression as a representation of the econometric model to be estimated, but the underlying theoretical arguments can be generalized to other econometric specifications, e.g., generalized linear models.\footnote{For example, a binary response model -- Probit or Logit -- can be formulated as a latent linear model with a binary transformation of the dependent variable, with the measurement error specified for the latent linear model.}

Consider a linear regression model,
\begin{equation}
\label{eq:true_reg}
Y=X\beta_X+\boldsymbol{Z\beta_Z}+\varepsilon,
\end{equation}
where $Y$ represents the dependent variable, $\{X,\boldsymbol{Z}\}$ represent the independent covariates ($\boldsymbol{Z}$ includes control variables and a constant term), $\varepsilon$ is the \textit{exogenous} random error term, and $\boldsymbol{\beta}=\{\beta_X,\boldsymbol{\beta_Z}\}$ denotes the model coefficients to be estimated. Importantly, $X$ is not directly observed in the data, and we instead rely on its surrogate, $\widehat{X}$, which is based on the \textit{predictions} of a machine learning model, e.g., random forest. For instance, if $X$ represents poverty level in a neighborhood, then $\widehat{X}$ could be the predicted poverty level mined from Google Street View images. In contrast, covariates $\boldsymbol{Z}$ are directly observed in the data and are measured precisely (with no measurement error). Therefore, the actual estimation being conducted is a regression of $Y$ on $\{\widehat{X},\boldsymbol{Z}\}$. In this section, we assume $X$ and $\widehat{X}$ to be continuous variables, whereas $\boldsymbol{Z}$ can contain both continuous and categorical variables. We discuss the case when $X$ and $\widehat{X}$ are binary variables in later sections. 

Because predictions of a machine learning model inevitably have some degree of error, $\widehat{X}$ is in general an imperfect surrogate for $X$, and contains continuous measurement error. The existence of measurement error in an independent covariate is known to result in biased regression estimates. As an illustrative example, consider the case of \textit{additive independent} (otherwise known as classical) measurement error, where $\widehat{X}=X+e$, and $e$ is independent of $X$. The estimated regression equation is therefore
\begin{equation}
\label{eq:est_reg}
Y=\widehat{X}\beta_X+\boldsymbol{Z\beta_Z}+(\varepsilon-e\beta_X).
\end{equation}
Because $Cov(\widehat{X},(\varepsilon-e\beta_X))=-\beta_X\sigma_e^2\neq0$, the regression suffers from endogeneity due to measurement error in $\widehat{X}$, resulting in biased coefficient estimates \citep{greene2003econometric}. (See Appendix \ref{FormalTheory} for a formal setup of this problem.)

\subsection{Building Random Forest}

Consider the task of building a random forest (or any predictive machine learning model) to predict $X$. A typical approach is to collect some amount of \textit{labeled data}, where the outcome to be predicted is actually observed. More concretely, denote the entire dataset that a researcher has access to as $D$. Suppose the researcher takes a random subsample from $D$, denoted as $D_{label}$, and obtains the (precisely-measured) ground truth for that subsample, e.g., via manual labeling. Then, $D_{label}$ would be randomly partitioned into $D_{train}$ and $D_{test}$, where $D_{train}$ would be used to build the random forest model and $D_{test}$ to evaluate the resulting model's performance. For data in the remaining unlabeled dataset, $D_{unlabel}=D \setminus D_{label}$, the random forest model would then be deployed to generate predictions $\widehat{X}$. 

As is common in a growing number of policy-relevant contexts, the dataset of interest often includes a large amount of unlabeled data. Due to the cost of obtaining ground truth labels, the size of $D_{label}$ is typically much smaller than the size of $D_{unlabel}$. As a result, researchers may not be able to estimate their econometric models of interest on $D_{label}$ with a satisfactory degree of statistical power given the likely large variance. That is, estimating the econometric models using information in $D_{unlabel}$ presumably has the potential to deliver substantial improvements in the precision of estimates (due to its larger sample size).

\subsection{Generating Candidate Instruments from Random Forest}

In this section, we consider the generation of instruments from random forest, which is at the core of the ForestIV approach. We first establish the setting and list assumptions, before providing formal theoretical results. We consider a random forest model with $M$ individual trees, indexed by $\{1,\dots ,M\}$, which is built based on $n$ training samples and $p$ features. On a new data point represented by feature vector $\mathbf{f} = \{f_1,\ldots,f_p\}$, denote the prediction of individual tree $i\in\{1,\dots ,M\}$ as $\widehat{X}^{(i)}$, and the prediction of the forest as $\widehat{X}$, where $\widehat{X} = \frac{1}{M} \sum_{i=1}^M \widehat{X}^{(i)}$. Given the ground truth, $X$, the prediction error is correspondingly defined as $e^{(i)}=\widehat{X}^{(i)} - X$. We make three assumptions underlying our theoretical results. The first two assumptions are adopted from \cite{scornet2015consistency}, which establishes the consistency of random forest predictions, and the third assumption is standard in the measurement error literature.

\textit{Assumption 1} \citep[Ground truth function,][]{scornet2015consistency}. The ground truth can be expressed as $X = \sum_{j=1}^p m_j(f_j) + \zeta$, where features $\{f_1,\ldots,f_p\}$ are uniformly distributed over $[0,1]^p$, $\zeta$ represents independent, centered Gaussian noise with finite variance, and each component $m_j(.)$ is continuous. This assumption states that the ground truth is the sum of univariate functions of input features. Although random forest is a nonparametric model, analysis of its properties is often facilitated within the framework of additive models \citep{scornet2015consistency}.

\textit{Assumption 2} \citep[Tree Growth,][]{scornet2015consistency}. Denote $t_n$ as the number of leaves in each tree, and $a_n$ as the number of training data points used to build each tree. Let $a_n \rightarrow \infty$, $t_n \rightarrow \infty$, we assume $t_n (\log a_n)^9/a_n \rightarrow 0$. This assumption, as a regularity condition, controls the rate at which the trees in the random forest grow.

\textit{Assumption 3} (Classical measurement error). As $n \rightarrow \infty$, the prediction error of an individual tree takes the classical form, i.e., $\lim_{n \rightarrow \infty} \mathbb{E}_i\mathbb{E}_{\mathbf{f}} Cov(X, e^{(i)}) = 0$. This assumption, notably common in the theoretical measurement error literature, implies that (asymptotically) the prediction error of an individual tree is uncorrelated with the ground truth \citep[e.g.,][]{hausman1991measurement,newey2001flexible,li2002robust,schennach2004estimation,schennach2013nonparametric}.

\begin{restatable}{theorem}{mainTheorem}
\label{mainTheorem}
Under Assumptions 1-3, for any two trees in the random forest, $i$ and $j$ ($i \neq j$),
$$\lim_{n \rightarrow \infty} \mathbb{E}_i\mathbb{E}_j\mathbb{E}_{\mathbf{f}} Cov(\widehat{X}^{(j)}, e^{(i)}) = 0.$$
\end{restatable}

This result, the proof of which can be found in Appendix \ref{Proofs}, implies that the expected covariance between one tree's prediction and another tree's prediction error goes to 0, which establishes an asymptotic guarantee for instrument validity and, therefore, the theoretical foundation for considering trees within a random forest as instruments for one another. To provide intuition for this foundation, we revisit the estimated regression
$Y=\widehat{X}\beta_X+\boldsymbol{Z\beta_Z}+(\varepsilon-e\beta_X)$, where $\widehat{X}$ is endogenous. A \textit{valid} instrument, $W$, should satisfy two conditions: (a) $Cov(W,\widehat{X}) \neq 0$, i.e., the instrument is correlated with the endogenous (mis-measured) regressor (i.e., the \textit{relevance} condition); and (b) $Cov(W,\varepsilon-e\beta_X)=0$, i.e., the instrument is not correlated with the regression error term (i.e., the \textit{exclusion} condition). Recall that $\varepsilon$ is exogenous in \eqref{eq:true_reg}, the true underlying regression equation of interest. Consequently, condition (b) is equivalent to $Cov(W,e)=0$, i.e., a valid instrument is not correlated with the measurement error in $\widehat{X}$. Now suppose we replace $\widehat{X}$ with $\widehat{X}^{(i)}$ (i.e., predictions of individual tree $i$ in the forest), and $W \equiv \widehat{X}^{(j)}$ (i.e., predictions of individual tree $j$ in the forest). We then have that $Cov(\widehat{X}^{(i)}, \widehat{X}^{(j)}) = Var(X) \neq 0$, meeting condition (a), and $Cov(e^{(i)}, \widehat{X}^{(j)}) \rightarrow 0$ as $n\rightarrow \infty$ (Theorem \ref{mainTheorem}), meeting condition (b). As a consequence, $\widehat{X}^{(j)}$ is an asymptotically valid instrument for $\widehat{X}^{(i)}$.

More generally, Theorem \ref{mainTheorem} tells us that, asymptotically, and under mild assumptions, in a random forest of $M$ trees, we can use predictions from any individual tree as the endogenous covariate in the regression model, and use predictions from the \textit{other} $M-1$ individual trees in the random forest as valid instruments. Though this result provides a valuable theoretical foundation for the random forest algorithm as a generator of valid instruments, we do not yet know at what sample size the asymptotic properties will manifest. Therefore, at a finite sample size we can only consider the individual trees as \textit{candidate} instruments, and we must carry out additional steps to identify which of these potential instruments provides evidence of having reached their asymptotic state. Moreover, the asymptotic result implies that with infinite data \textit{all} individual trees are valid instruments for all other trees. However, in practice, only one pair of (endogenous and valid instrumenting) trees is required, which (given the asymptotic result) would appear plausible even in finite data.

We note that, in finite samples, the above asymptotic results with respect to random forest are in fact also reflected in empirical evidence. For example, \cite{bernard2010study} show that the performance of random forest statistically improves with increases in the accuracy of individual trees and decreases in the \textit{correlation} between their prediction errors. In other words, a well-performing random forest should consist of individual trees that are relatively accurate (high \textit{strength}) and have only weakly correlated errors (low \textit{correlation}) \citep{Breiman2001}. We make the key observation that high strength and low correlation are closely related to the requirements of a valid instrumental variable.

This highlights an interesting and perhaps counter-intuitive characteristic of ForestIV. Because the predictive performance of an individual tree is typically worse than that of the entire random forest, we likely induce larger estimation biases by drawing on predictions from an individual tree as the endogenous covariate (rather than the aggregate prediction from the overall forest). However, this initial sacrifice is accompanied by the opportunity to leverage predictions from other trees as instruments, to address the estimation bias. In other words, for the purposes of \textit{causal inference}, a more nuanced use of the entire random forest ensemble (rather than a traditional use of its aggregate predictions) allows the mitigation of the measurement error that is inevitably present in the model's predictions.

\subsection{Selecting Instruments for Correction}
In practice, when using the predictions from a focal individual tree in a random forest as the endogenous covariate, it is necessary to select only a \textit{subset} of the other trees' predictions for use as instruments, for several reasons. First, when there is only a finite set of training data, it is unlikely that all other trees have reached their (asymptotic) state of valid instrumentation for the focal tree. Second, due to randomness in constructing the random forest with finite data, even valid instruments can be invalid by chance, i.e., empirically $Cov(e^{(i)}, \widehat{X}^{(j)})$ is large. Third, using an excessive number of instruments can lead to over-fitting of the endogenous variables, including the endogenous components a researcher seeks to eliminate, and create estimation challenges \citep{roodman2009note}. Finally, again due to randomness in the construction of a random forest, some instruments may be only weakly correlated with the endogenous covariate, despite meeting the exclusion requirement. Including these \textit{weak instruments} in an instrumental variable regression can be counterproductive, yielding biased and inconsistent estimations \citep{hausman2001mismeasured}. 

In our setting, because we have access to a set of labeled data, we can assess instrument validity empirically, to a degree. Thus, after obtaining the predictions from each individual tree in the random forest, we focus on identifying a ``desirable" subset of trees: one to be used as the endogenous covariate, and the remaining to serve as valid and strong instruments that can mitigate estimation biases. We decompose this task into three distinct steps, as follows:

\begin{itemize}

    \item \textbf{Step 1 Removal of Invalid Instruments}: Given $i\in\{1,\ldots, M\}$, use $\widehat{X}^{(i)}$ as the endogenous covariate, then select a subset of other trees, $\boldsymbol{V_i} \subseteq \{\widehat{X}^{(1)}, \ldots, \widehat{X}^{(M)}\} \setminus \widehat{X}^{(i)}$, which omit \textit{invalid} instruments for $\widehat{X}^{(i)}$. This step is conducted using $D_{test}$.

    \item \textbf{Step 2 Selection of Strong Instruments}: Given $i\in\{1,\ldots, M\}$, use $\widehat{X}^{(i)}$ as the endogenous covariate, then select a subset of other trees, $\boldsymbol{S_i} \subseteq \boldsymbol{V_i}$, that consists of \textit{strong} instruments for $\widehat{X}^{(i)}$. This step is conducted using $D_{test} \cup D_{unlabel}$. Steps 1-2 are iterative (discussed in detail later). 
  
    \item \textbf{Step 3 Estimation}: Based on selected instruments $\boldsymbol{S_i}$ for covariate $\widehat{X}^{(i)}$, obtain the 2SLS regression estimates. We carry out additional checks to gauge the validity of 2SLS regression estimates, and retain the 2SLS estimates that meet the specific checks to produce the final corrected coefficient estimates. This step is conducted using $D_{label} \cup D_{unlabel}$.
  
\end{itemize}

The additional ``validity checks" in Step 3 are necessary because, even though Steps 1-2 are designed to remove invalid instruments and select strong ones, there is limited theoretical guarantee that all selected instruments are always valid and strong in \textit{finite samples}. The validity checks in the third step thus attempt to identify for which subset of trees the asymptotic properties appears to be present, reducing the likelihood that our approach produces erroneous results when predictions from all individual trees in the random forest are not suitable instruments. Moreover, with a finite sample, it is unlikely that instrument validity, particularly the exclusion restriction, will be satisfied \textit{exactly}. Therefore, our procedure is consistent with the practice of using ``plausibly exogenous" instruments for estimation in finite samples \citep{conley2012plausibly}.

\subsubsection{Step 1: Removal of Invalid Instruments}

To rule out invalid instruments for a given endogenous covariate, we rely on information from the labeled data. Recall that the random forest model is built on $D_{train}$, and its performance is then evaluated on $D_{test}$. As a result, for $D_{test}$, we observe the ground truth, the model-predicted values, and thus the prediction errors (the difference between ground truth and prediction). Using this information on $D_{test}$, we can gauge the validity of using individual trees' predictions as instruments, and empirically rule out the invalid instruments that are strongly correlated with the measurement errors. 

Grounded in prior work on instrument selection \citep[e.g.,][]{Belloni2012}, we adapt a lasso-based heuristic procedure to identify and discard instruments violating the exclusion requirement. Without loss of generality, suppose predictions from the first individual tree on $D_{test}$, $\widehat{X}^{(1)}$, serve as the endogenous covariate, and the corresponding prediction error is denoted as $e^{(1)}$. We estimate a lasso regression of $e^{(1)}$ on the predictions of other individual trees on $D_{test}$, i.e., $e^{(1)} \sim \{\widehat{X}^{(2)}, \dots, \widehat{X}^{(M)}\}$. The set of regressors for which the lasso yields non-zero coefficients are then dropped, because their linear combination is determined to be a strong predictor of $e^{(1)}$, which implies that those regressors violate the exclusion restriction with respect to $\widehat{X}^{(1)}$. Conversely, the set of regressors with zero coefficients are retained, denoted by $\boldsymbol{V_1}$, because the lasso fails to provide evidence to suggest that they violate the exclusion restriction. 

\subsubsection{Step 2: Selection of Strong Instruments}

To select a set of sufficiently strong instruments from the remaining set of candidate instruments, for a given endogenous covariate, we adopt another heuristic approach motivated by \cite{Belloni2012}. Consider $\widehat{X}^{(1)}$ as the endogenous covariate, and $\boldsymbol{V_1}$ as the set of instruments obtained in Step 1. We estimate a lasso regression in which the endogenous covariate is regressed on all available instruments (similar to the first stage of a 2SLS estimation), i.e., $\widehat{X}^{(1)} \sim \boldsymbol{V_1}$. The lasso attempts to shrink the coefficients of instruments that are conditionally uncorrelated with the endogenous covariate (i.e., weak in the presence of the other instruments) to zero, with the non-zero coefficients indicating a set of regressors predictive of $\widehat{X}^{(1)}$, which intuitively correspond to \textit{strong} instruments. Because this step requires only the predictions from individual trees, it is carried out on $D_{test} \cup D_{unlabel}$.

Importantly, if the set of instruments changes during the two-step selection process (i.e., certain instruments are determined to violate exclusion or relevance requirements, and are thus dropped), we will \textit{repeat} both of the lasso selection steps (1 and 2) with the remaining instruments, until the selection ceases to change. This iterative approach increases the likelihood that our selected instruments are simultaneously valid and strong. Moreover, we expect this procedure should work well with sufficient data, because when $\boldsymbol{V_i}$ contains only excluded instruments (satisfied asymptotically based on Theorem \ref{mainTheorem}), \cite{Belloni2012} shows that asymptotically $\boldsymbol{S_i}$ will be the linearly optimal set of instruments. We offer greater detail about the instrumental variable selection procedure outlined in Steps 1 and 2 in Appendix \ref{Algorithm1}. In general, for each $\widehat{X}^{(i)}, i \in \{1, \dots, M\}$, we can use this procedure to select a set of strong, excluded instruments, $\boldsymbol{S_i}$. 

\subsubsection{Step 3: Estimation}
Consider the econometric model of interest, $Y = X\beta_X + \boldsymbol{Z\beta_Z} + \varepsilon$. If all variables in the econometric model, i.e., $\{Y,X,\boldsymbol{Z}\}$, are directly observable in $D_{label}$, then one can obtain unbiased estimates of $\{\beta_X,\boldsymbol{\beta_Z}\}$, denoted as $\widehat{\boldsymbol{\beta}}_{label}$, by estimating the econometric model on the error-free $D_{label}$. In practice, the size of $D_{label}$ is often limited, due to the cost of acquiring labels. As a result, $\widehat{\boldsymbol{\beta}}_{label}$ may exhibit a large standard error and would not be particularly suitable for drawing causal inferences. Nonetheless, the estimates are unbiased and can be employed as a useful baseline for determining the quality of instrument estimations.\footnote{We note that if $D_{label}$ is sufficiently large, implying that researchers can afford to acquire a large set of labeled data with sufficient statistical power to estimate the econometric model of interest, then it is unnecessary to use machine learning models to construct / mine variables in the first place; the researcher should simply estimate the econometric model on $D_{label}$.} In this sub-section, we discuss the final selection procedure, based on a comparison of 2SLS estimates and the unbiased estimation using precisely measured covariates in $D_{label}$, i.e., $\widehat{\boldsymbol{\beta}}_{label}$.

Specifically, following Steps 1-2, we obtain a set of 2SLS estimates for each pair of $(\widehat{X}^{(i)}, \boldsymbol{S_i})$, using $\widehat{X}^{(i)}$ as the endogenous covariate, and $\boldsymbol{S_i}$ as its instruments, denoted as $\widehat{\boldsymbol{\beta}}_{IV}^i$. Denote the variance-covariance matrices of $\widehat{\boldsymbol{\beta}}_{IV}^i$ and $\widehat{\boldsymbol{\beta}}_{label}$ as $\widehat{\boldsymbol{\Sigma}}_{IV}^i$ and $\widehat{\boldsymbol{\Sigma}}_{label}$, respectively. To compare $\widehat{\boldsymbol{\beta}}_{IV}^i$ with $\widehat{\boldsymbol{\beta}}_{label}$, we use the Hotelling's $T^2$ test with unequal variance \citep{seber2009multivariate}. This test is a multivariate generalization of the $T$ test, designed to evaluate a null hypothesis of mean equality between two vectors of random variables, whose joint distributions have unequal variances. The test statistic is:
$H_i = (\widehat{\boldsymbol{\beta}}_{IV}^i - \widehat{\boldsymbol{\beta}}_{label})^T \left( \widehat{\boldsymbol{\Sigma}}_{IV}^i + \widehat{\boldsymbol{\Sigma}}_{label} \right)^{-1} (\widehat{\boldsymbol{\beta}}_{IV}^i - \widehat{\boldsymbol{\beta}}_{label})$. 
$H_i$ is asymptotically distributed as $\chi^2(K)$, where $K$ represents the total number of covariates \citep{seber2009multivariate}. Therefore, if $H_i$ is larger than the critical value of $\chi^2(K)$ at a user-chosen significance level (we pick $\alpha=0.05$ for illustration purposes), then $\widehat{\boldsymbol{\beta}}_{IV}^i$ is significantly different from the unbiased $\widehat{\boldsymbol{\beta}}_{label}$, indicating that one or more instruments likely violate the relevance and/or exclusion requirement. Such $\widehat{\boldsymbol{\beta}}_{IV}^i$ is then discarded.

Meanwhile, suppose more than one $(\widehat{X}^{(i)}, \boldsymbol{S_i})$ pair has an associated Hotelling statistic below the critical value, such that each pair empirically yields estimates that are not significantly different from $\widehat{\boldsymbol{\beta}}_{label}$. We then measure the empirical bias and variance of each set of estimates, using \textit{empirical MSE}, as follows: $$MSE_i = tr\left((\widehat{\boldsymbol{\beta}}_{IV}^i - \widehat{\boldsymbol{\beta}}_{label}) (\widehat{\boldsymbol{\beta}}_{IV}^i - \widehat{\boldsymbol{\beta}}_{label})^T + \widehat{\boldsymbol{\Sigma}}_{IV}^i\right)$$. Finally, we select the estimates that have the smallest empirical MSE, i.e., the estimates that have the smallest sum of empirical bias and variance. We note that testing these multiple hypotheses, without control, will potentially inflate our Type I error: the probability of incorrectly rejecting a truly valid tuple. Given that we likely have multiple tuples to choose from, our preference is to be conservative (over-reject tuples) and we prefer to make \textit{no} inference rather than make an erroneous inference.

Note that the 2SLS variances (i.e., diagonal elements of $\widehat{\boldsymbol{\Sigma}}_{IV}^i$) associated with ForestIV estimates do not reflect the estimator's true variability, because we intentionally select estimates with low variances. To properly quantify the variability of ForestIV estimator, we propose to bootstrap $\{D_{label},D_{unlabel}\}$, apply the ForestIV procedure to each bootstrap sample, and obtain the empirical sampling distributions of ForestIV estimates. The variances of the sampling distributions serve as straightforward estimation of ForestIV variances. The pseudocode for ForestIV is summarized in the Appendix \ref{Algorithm2}.

Taking the proposed ForestIV procedure as a baseline, we also consider three possible alternative design choices, related to instrument discovery, selection, and estimation. First, a key feature of ForestIV is that we ``break up" a random forest and exploit individual trees to discover instruments for bias correction. As an alternative, one might instead consider a \textit{sample splitting} approach, where $D_{train}$ is split into two independent subsets and a random forest is built on each subset. Then, predictions from one forest might serve as an instrument for predictions from the other. This approach is intuitively appealing because the predictive performance of a forest is arguably better than that of a single tree, thereby reducing the measurement error problem to start with. In the same vein, a second alternative design of the instrument selection procedure in ForestIV is to use the aggregate predictions from a \textit{subset of trees} (rather than predictions from a single tree) as the endogenous covariate and instrumental variables. Third, in the estimation step of ForestIV, while we propose to select the estimates that minimize empirical MSE, an alternative design is to \textit{average} across all estimates that are not rejected by the Hotelling's $T^2$ test. In general, the relative utility of these alternative design choices is not clear a priori, thus we explore all of them empirically as part of the simulation experiments discussed in the following section.

\section{Simulation Experiments}

\subsection{Basic Simulation Setup}
Our first demonstration leverages the Bike Sharing Data \citep{fanaee2014event}, which contains 17,379 instances of hourly bike rental activities. It is a commonly used dataset for benchmarking and evaluating machine learning models \citep[e.g.,][]{giot2014predicting}. Moreover, this dataset assist in emulating an example of how machine learning could enable empirical studies of questions on the impact of mobility on important economic or social outcomes.

We randomly partition the dataset, such that 1,000 observations serve as $D_{train}$, 200 observations serve as $D_{test}$, and the remaining 16,179 observations serve as $D_{unlabel}$. This represents a realistic scenario where $D_{unlabel}$ is much larger than $D_{label}$. Using $D_{train}$, we build a random forest model of 100 trees to predict the log-transformed count of total hourly bike rentals (denoted as $lnCnt$, log-transformed to reduce skewness), based on 12 features, including the time of the rental as well as weather and seasonal information. Importantly, the random forest generates aggregate (ensemble) predictions, as well as predictions from each of its individual trees. We denote $\widehat{lnCnt}$ as the aggregate predictions, and $\widehat{lnCnt}_i$ as the predictions from individual tree $i \in \{1,\ldots,100\}$.

Next, we simulate an econometric model with $lnCnt$ as an independent covariate. The model specification is $Y=1+0.5lnCnt+2Z_1+Z_2+\varepsilon$, where $Z_1 \sim Uniform[-10,10]$, $Z_2 \sim N(0,100)$, and $\varepsilon \sim N(0,4)$. In the Bike Sharing dataset, $\sigma_{lnCnt}=1.5$, which is smaller than the standard deviation of the regression error term. Therefore, this simulation setup represents a realistic scenario where the variation in ``noise" is comparable to that of the ``signal". We keep the true regression coefficients fixed, i.e., $\beta_0=1, \beta_{lnCnt}=0.5, \beta_{Z_1}=2, \beta_{Z_2}=1$, in order to quantify the degree of estimation bias and the effectiveness of any corrections.

The above simulation procedure is repeated for 100 rounds. Within each simulation round, a random forest model is built based on a random split of the data (as described earlier), and an artificial dataset is generated for the econometric estimations. Specifically, we first estimate the \textit{biased regression}, $Y \sim \beta_0+\beta_1 \widehat{lnCnt}+\beta_2Z_1+\beta_3Z_2$, where $\widehat{lnCnt}$ is the aggregate predictions from that round's random forest. This is precisely what would be done if one were to use the machine-learning-predicted covariate directly in the econometric model, without considering measurement error, as is commonly the case. We then apply ForestIV to obtain the corrected coefficients. Because an independent set of $\{D_{train}, D_{test}, D_{unlabel}\}$ is drawn in each simulation round, the ForestIV estimates across 100 simulation rounds naturally form the empirical sampling distributions. We report the standard deviations of the sampling distributions as the standard errors of ForestIV estimates.\footnote{We also obtained standard errors from our bootstrap approach, and note that the estimates are nearly identical.} Finally, besides the biased and ForestIV estimates, we also report the \textit{unbiased estimates} from directly running the regression on $D_{label}$.

\subsection{Basic Simulation Results}
We provide some descriptive evidence regarding the effectiveness of our two-step lasso regression procedure to select valid and strong instruments from candidates in Appendix \ref{Bike_Selection}. Next, we plot the sampling distributions of biased, unbiased, and ForestIV estimation on $lnCnt$ (i.e., the covariate generated by machine learning) across all 100 simulation runs in Figure \ref{fig:BikeData}. We report complete results in Appendix \ref{Bike_MainResults}.

\begin{figure}[tbhp]
\label{fig:BikeData}
\centering
\includegraphics[width=0.5\linewidth]{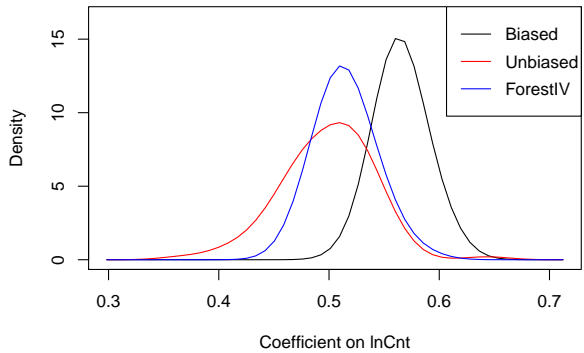}
\caption{Distributions of Biased, Unbiased, and ForestIV Estimation on $lnCnt$ across 100 Simulation Runs}
\end{figure}

We make several observations. First, directly using the $\widehat{lnCnt}$ predicted by random forest in the regression model clearly results in biases. The coefficient on $lnCnt$ is \textit{overestimated} on average. Second, compared to the biased regression, ForestIV effectively mitigates estimation bias on $lnCnt$. Third, compared to the unbiased estimation, ForestIV estimation has a ``narrower" distribution (i.e., has smaller standard errors), indicating sizable increase in estimation precision.

We conduct several sensitivity analyses to understand performance of ForestIV with respect to several parameters. We summarize the key insights from the sensitivity analyses as follows, while referring to Appendix \ref{Bike_Sensitivity}. for detailed results. First, we examine how ForestIV estimations change as the \textit{size of the unlabeled dataset} increases (Appendix \ref{Bike_Sensitivity_Size}). This helps illustrate the asymptotic properties of ForestIV estimations. We found that ForestIV estimates converge as more unlabeled data is added, with shrinking confidence intervals. The converged estimates substantially reduce bias and achieve smaller standard errors than the unbiased estimates. Second, we further increase the amount of noise in the data used for econometric estimation, and check the corresponding performance of ForestIV (Appendix \ref{Bike_Sensitivity_Noise}). While standard errors of ForestIV estimates become larger, we still observe effective bias correction. Third, we vary the total number of trees in the random forest model, and found that ForestIV requires a reasonably large random forest to perform well (Appendix \ref{Bike_Sensitivity_NTree}). Moreover, the Hotelling $T^2$ statistics turn out to be a useful metric to help determine the size of random forest.

Finally, we implement and empirically evaluate the relative benefit of the three alternative design choices noted earlier: (i) sample splitting (Appendix \ref{Bike_AlternativeDesign_SampleSplit}), (ii) using predictions from subsets of trees as the endogenous covariate and instruments (Appendix \ref{Bike_AlternativeDesign_SubsetTree}), and (iii) averaging across all ForestIV estimates not rejected by the Hotelling $T^2$ test (Appendix \ref{Bike_AlternativeDesign_Averaging}). We refer readers to the appendices for a detailed discussion of each simulation-based evaluation. On the whole, we find that none of these alternative designs outperforms our proposed ForestIV method (at least in the context of our Bike Sharing data and associated simulation setup). That said, the possibility remains that one or more of these alternatives could yield improvements under certain conditions, and each holds promise as an avenue of further exploration in future work. 

\section{ForestIV with Binary Endogenous Covariates}

In this section, we turn to the case of a binary misclassified (and thus endogenous) covariate, as would be generated by a machine-learning classifier model. It turns out that ForestIV still exhibits an ability to generate (and select) instrumental variables that produce improved estimates, although the underlying mechanisms are somewhat different from the case of a continuous endogenous covariate.

\subsection{Theoretical Results}

Suppose the binary outcome label can have values of 0 (negative class) or 1 (positive class). Consider a random forest \textit{classifier} with $M$ decision trees. We again use notations $X$, $\widehat{X}$, and $\widehat{X}^{(i)}$ to represent the ground truth, prediction of the forest, and prediction of individual tree $i$. The prediction error of individual tree $i$ is defined as $e^{(i)} = \widehat{X}^{(i)} - X$. For any given data point, $e^{(i)}$ can take three possible values: 0 (correct prediction), 1 (false positive), and $-1$ (false negative). 

Meanwhile, based on the econometric literature, e.g., \cite{angrist2008mostly}, instrumental variables with binary endogenous covariates can be applied in the same manner as in the continuous case; that is, treat the variable as continuous, and employ a 2SLS estimation. The first stage of this 2SLS estimation amounts to a linear probability model. The predicted values for the endogenous covariate recovered from the first stage regression thus reflect exogenous variation in continuous class probabilities. These values would then be used in the second stage regression. Intuitively, to evaluate whether predictions from one tree, $j$, can serve as a valid instrument for predictions from another tree, $i$, we need to assess (1) $Cov(\widehat{X}^{(i)}, \widehat{X}^{(j)})$ (the relevance condition) and (2) $Cov(e^{(i)}, \widehat{X}^{(j)})$ (the exclusion restriction).

The first condition is typically satisfied, because two individual trees from a reasonably well-performing random forest are both somewhat predictive of the outcome, i.e., $Cov(\widehat{X}^{(i)}, \widehat{X}^{(j)})\neq 0$. The second condition can be written as $Cov(e^{(i)}, e^{(j)} + X)$. Accordingly, we next provide several theoretical results to characterize $Cov(e^{(i)}, e^{(j)})$ and $Cov(e^{(i)}, X)$, respectively.

\begin{restatable}{theorem}{theoremBinaryA}
\label{theoremBinaryA}
The error rate of a random forest binary classifier decreases with $\mathbb{E}_j \mathbb{E}_i Corr(|e^{(i)}|, |e^{(j)}|)$, where $e^{(i)}$ and $e^{(j)}$ are prediction errors of tree $i$ and tree $j$ ($i\neq j$).
\end{restatable}

All proofs are included in Appendix \ref{Proofs}. This theorem suggests that a well-performing random forest would have relatively small $Corr(|e^{(i)}|, |e^{(j)}|)$. Typically, because $Corr(|e^{(i)}|, |e^{(j)}|) \neq Corr(e^{(i)}, e^{(j)})$, this result seems to indicate that the exclusion restriction of instrument validity, $Cov(e^{(i)}, e^{(j)} + X) = 0$, is not satisfied. However, the following two theorems show that $Cov(e^{(i)}, X)$ is always non-zero (i.e., the classical measurement error assumption is never true for binary misclassification), and can be offset, or cancelled out, $Cov(e^{(i)}, e^{(j)})$, thereby making the exclusion condition plausible.

\begin{restatable}{theorem}{theoremBinaryB}
\label{theoremBinaryB}
$\forall i \in \{1,\ldots,M\}$, $Cov(e^{(i)}, X) < 0$.
\end{restatable}

For notation simplicity, denote the probability that $X = \alpha, \widehat{X}^{(i)} = \beta, \widehat{X}^{(j)} = \gamma$ as $p_{\alpha \beta \gamma}$ ($\alpha, \beta, \gamma \in \{0,1\}$), and denote the probability that $X = \alpha$ as $p_{\alpha \bullet \bullet}$.

\begin{restatable}{theorem}{theoremBinaryC}
\label{theoremBinaryC}
$\forall i \neq j \in \{1,\ldots,M\}$, $Cov(e^{(i)}, e^{(j)}) > 0$ if and only if $(p_{000}+p_{111})(p_{011}+p_{100}) + 2(p_{0 \bullet \bullet}-p_{000})p_{100} + 2(p_{1 \bullet \bullet}-p_{111})p_{011} + (p_{010}-p_{101})(p_{110}-p_{001}) > 0$.
\end{restatable}

Theorem \ref{theoremBinaryB} shows that the prediction error of an individual tree is always negatively correlated with the ground truth. Theorem \ref{theoremBinaryC} suggests that, a few corner cases aside, the correlation between prediction errors of two individual trees is positive.\footnote{One such corner case that $Cov(e^{(i)}, e^{(j)}) < 0$ is when $p_{011}=p_{100}=0$ and $(p_{010}-p_{101})(p_{110}-p_{001}) < 0$, which means that the predictions from the two trees are never wrong at the same time, and their misclassification patterns satisfy a stringent condition. Such corner case is not easily realized.} As a result, $Cov(e^{(i)}, X)$ is likely to offset $Cov(e^{(i)}, e^{(j)})$, leading to a relatively small value of $Cov(e^{(i)}, e^{(j)} + X)$.\footnote{This is empirically supported in our simulation experiments using the Bank Marketing data, discussed in the next section. In particular, the average covariance between prediction errors of any two trees, $\mathbb{E}_{i\neq j}Cov(e^{(i)}, e^{(j)})$, is 0.071 (correlation is 0.465). Meanwhile, the average covariance between a tree's prediction errors and the ground truth, $\mathbb{E}_i Cov(e^{(i)}, X)$, is -0.079 (correlation is -0.601). As a result, the average covariance between one tree's predictions and another tree's prediction errors, $\mathbb{E}_{i\neq j}Cov(e^{(i)}, \widehat{X}^{(j)})$, is only -0.008 (correlation is -0.061).} In other words, in the case of a binary endogenous covariate generated by a random forest classifier, other trees' predictions can still plausibly serve as instrumental variables.

\subsection{Simulation Experiments}

We demonstrate the performance of ForestIV in the case of binary classification using the Bank Marketing data \citep{moro2014data} as an example dataset, which contains 45,211 records related to a bank's telemarketing efforts. We randomly partitioned the data into 1,500 observations to serve as $D_{train}$, 500 observations to serve as $D_{test}$, and the remaining 43,211 observations to serve as $D_{unlabel}$. Using the training data, we build a random forest classifier comprised of 100 trees to predict a binary outcome, $Deposit$, representing whether a client subscribed to a term deposit as a result of the phone call, based on 16 attributes describing the client and the marketing campaign. This dataset is an example of how machine learning can enable empirical studies of questions in direct marketing or customer relationship management.

We next simulate an econometric model: $Y=1+0.5Deposit+2Z_1+Z_2+\varepsilon$, where $Z_1 \sim Uniform[-1,1]$, $Z_2 \sim N(0,1)$, and $\varepsilon \sim N(0,4)$. As before, we repeat the simulation for 100 rounds. Within each round, we estimate the biased regression (directly using the random forest predictions, $\widehat{Deposit}$, in the regression), the unbiased regression obtained on $D_{label}$, and the corrected coefficient obtained from the ForestIV procedure on $D_{unlabel}$. We plot the distributions of biased, unbiased, and ForestIV estimation on $Deposit$ in Figure \ref{fig:BankData}, and report the full results in Appendix \ref{Bank_MainResults}.

\begin{figure}[tbhp]
\label{fig:BankData}
\centering
\includegraphics[width=0.5\linewidth]{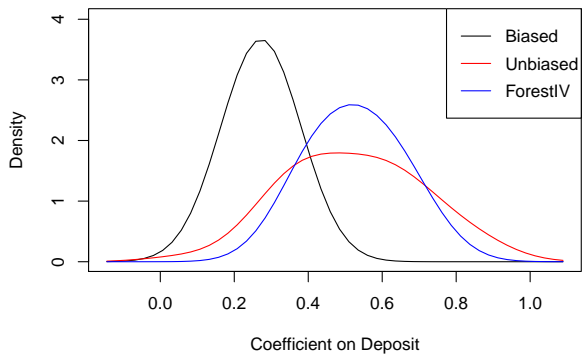}
\caption{Distribution of Biased, Unbiased, and ForestIV Estimation on $Deposit$ across 100 Simulation Runs}
\end{figure}

Consistent with previous simulations, we observe that directly using the random forest's predictions in our regression leads to severe biases. The coefficient on $Deposit$ is underestimated. ForestIV is again effective in mitigating the biases. Finally, compared to unbiased estimation, ForestIV achieves gains in estimation precision, as indicated by its ``narrower" distribution. In general, this set of simulations confirms the effectiveness of ForestIV for a binary endogenous covariate. We repeat all the sensitivity analyses that were conducted previously in the continuous covariate scenario and observe consistent insights. All results are described in Appendix \ref{Bank_Sensitivity}.

In addition to the above simulation experiments, we conduct benchmarking analyses to compare the relative correction performance of ForestIV with those of three other approaches from the existing literature: Simulation-Extrapolation \citep[SIMEX,][]{Cook1994}, Latent Instrumental Variables \citep[LatentIV,][]{ebbes2005solving}, and regression adjustment for nonparametrically generated regressors \citep{meng2016linear}. First, SIMEX is an alternative correction approach that has previously been demonstrated to perform well in the same problem context. We run two sets of simulations to show that ForestIV achieves better correction performance than SIMEX. We also identify a systematic ``blindspot" in the SIMEX correction procedure that occurs when measurement error is correlated with other precisely-measured covariates in an econometric model. While SIMEX will produce provably erroneous correction results in the presence of such correlation, empirical results suggest that our ForestIV approach does not suffer from this limitation. Second, with the LatentIV approach, a latent instrumental variable is modeled and jointly estimated with the main regression to address the endogeneity issue. Our simulation again shows that ForestIV outperforms LatentIV in bias correction. Third, in the generated regressors literature, \cite{meng2016linear} have proposed a method to explicitly adjust coefficient estimates to account for biases due to measurement error in linear regressions, deriving from mis-measured, nonparametrically generated regressors. Once more, our simulation shows that ForestIV outperforms this method in bias correction. Thus, in sum, ForestIV is found to outperform all three benchmarks in our simulation experiments (see Appendix \ref{Benchmarking} for details).

\section{Conclusion and Future Work}

To summarize, we introduce a new approach, ForestIV, which addresses bias in regression estimates that is attributable to (predictive) measurement error in data-mined covariates. With a continuous endogenous covariate, the intuition behind ForestIV is that a high-performing random forest will be comprised of (i) trees that are individually accurate in their predictions and, thus, ``overlap", offering ``repeated measures" of true, exogenous variation in the data-mined variable, and (ii) trees that exhibit low correlations in their prediction errors, which in tandem with the former point implies that trees make ``different" mistakes, and thus embed orthogonal measurement errors. Our approach is closely connected to the idea of using multiple error-prone measures as instrumental variables \citep[e.g.,][]{blackburn1992unobserved,hausman1995nonlinear,lewbel2019using}. With a binary endogenous covariate, while trees no longer necessarily have low error correlations, we show that instrument validity is nonetheless plausible.

The application of our approach in empirical contexts has the potential to improve the precision and robustness of estimations, and thus subsequent decision-making. Meanwhile, our approach shows the possibility of automatically generating candidate instruments based on an ensemble learning technique, which complements the emerging literature on the use of machine learning methods for causal inference \citep[e.g.,][]{athey2016recursive,mcfowland2018efficient}. At the core of ForestIV is the fundamental trade-off between an estimator's bias and variance, which together describe its statistical risk. ForestIV attempts to provide an estimate with reduced overall risk, one with substantially lower bias than the biased regression and lower variance than the unbiased regression. We provide several additional guidelines for using ForestIV in practice in Appendix \ref{Practice}.

Several future research directions are worth pursuing. For example, while this paper focuses on selecting valid and strong instruments from a given random forest, future work might look to leverage a specialized random forest algorithm that explicitly aims to minimize individual trees' prediction error correlations. The Rotation Forest algorithm \citep{blaser2016random} and the Dynamic Random Forest algorithm \citep{bernard2012dynamic} represent two such attempts. Another direction to pursue is to generalize ForestIV to bagging-based machine learning models, more broadly. Intuitively, because individual learners from a bagging model are trained on different bootstrap samples of the training data, they are likely to generate correlated predictions and weakly correlated prediction errors. Future work can investigate whether general bagging models can produce useful instruments, whether the types of individual learners (e.g., decision tree or other techniques) affect the validity of instruments, and the performance of those alternative ensembles relative to ForestIV.

\clearpage
\appendix
\section*{Appendix}

\section{Relevant Literature} \label{Literature}

The problem of measurement error has been studied extensively in the econometrics literature. In regression models, measurement error in independent covariates is a form of endogeneity \citep{greene2003econometric} and is known to lead to biased coefficient estimates, not only for the mis-measured covariate, but also for coefficients associated with other, precisely measured covariates appearing in the same regression (unless the precisely measured covariates are strictly independent of the measurement error). In contrast to the common (mis-held) belief that measurement error only leads to attenuation of coefficient on the mis-measured covariate (i.e., bias toward zero), the actual direction of bias is difficult to anticipate, particularly as the econometric specification or the structure of the measurement error grow more complicated \citep{gustafson2003measurement,schennach2016recent,yang2018mind}. In general, ignoring measurement error may lead to errors in sign, magnitude, and statistical significance of coefficient estimates.

\subsection{Instrumental Variable Approach}

\textit{Instrumental variables} are a standard approach to addressing the measurement error problem; they can be used in a two-stage least-squares estimation to mitigate associated estimation biases. A valid instrument in this case will be correlated with the mis-measured covariate but not its measurement error. However, researchers using the instrumental variable approach often face two significant challenges. First, valid instruments are not easy to locate. Typically, instruments cannot be identified absent knowledge of the underlying data-generation process, i.e., the nature of the endogeneity at play.\footnote{As an example, see https://thepoliticalmethodologist.com/2015/07/13/why-cant-we-just-make-up-instrumental-variables/ for a discussion of why purely randomly generated instruments are typically invalid.} Second, to justify the validity of a proposed instrument, researcher needs to provide convincing evidence that the instrument satisfies two criteria, namely \textit{relevance} and \textit{exclusion}. While the former criterion can be evaluated by empirically examining strength of the association between the endogenous covariate and the instrument, the latter is often untestable, and thus depends on the researcher offering a convincing qualitative, conceptual argument that the instrument has no association with the final outcome of interest, except via its influence on the endogenous variable. 

Our algorithm offers a novel opportunity to achieve both requirements, through quantitative means, as it leverages the availability of a perfectly measured set of data, i.e., the predictive model's labeled data, reducing the need for qualitative arguments (or restrictive assumptions).

\subsection{Ensemble Learning and Random Forest}

In the machine learning literature, ensemble learning represents an important paradigm in the formulation of predictive models. Instead of building one model to solve a prediction problem, ensemble learning aims to build multiple individual models, i.e., an ensemble of models, and to combine their individual predictions to arrive at a more accurate and stable aggregate prediction \citep{aggarwal2015data}. Some typical ensemble learning methods include bagging \citep{breiman1996bagging}, boosting \citep{freund1996experiments}, and random forest \citep{Breiman2001,denisko2018classification}, of which the latter is particularly relevant to our paper. A random forest is an ensemble of decision trees. Each tree is built on a random sample of the training data, and a random subset of features are considered for each split (node) in a tree \citep{Breiman2001}. The forest's prediction for an observation results from aggregating predictions from individual trees, e.g., majority voting for classification tasks or averaging for numeric prediction tasks.

Random forests have proven extremely useful in a variety of fields of study, due to their commonly high predictive accuracy \citep{verikas2011mining,denisko2018classification}. The predictive performance of a random forest is positively associated with the accuracy of each individual tree and negatively associated with inter-tree correlations in prediction errors \citep{Breiman2001,bernard2010study}. Intuitively, the performance of a random forest increases as a joint function of the individual prediction accuracy of trees comprising the forest, and the degree to which constituent trees make "different" prediction errors. Based on the observation that these objectives are closely analogous to the relevance and exclusion restrictions that underpin valid instrumental variables, there is some face validity to the idea that these individual trees may serve as candidate instruments for one another to resolve endogeneity in later regressions arising from predictive (measurement) error.

\subsection{The Generated Regressors Literature}

The measurement error problem we study here is closely related to a large body of econometrics literature on ``generated regressors", where certain covariates in econometric estimations are not directly observed; rather, they are first estimated. In fact, two-stage least-square (2SLS) estimation with instrumental variables is one such generated regressor model, where predicted values of the endogenous covariate used in the second stage regression are generated from the first stage regression. Researchers have examined the theoretical properties of econometric models incorporating generated regressors in cases where the generating function or the final estimation are parametric \citep[e.g.,][]{newey1984method,murphy1985estimation}, semi-parametric \citep[e.g.,][]{blundell2004endogeneity,mammen2016semiparametric}, or non-parametric \citep[e.g.,][]{sperlich2009note,mammen2012nonparametric}. For in-depth reviews of this extensive literature, we refer the reader to \cite{pagan1984econometric} and \cite{oxley1993econometric}. In some of this work, researchers have noted that the generating function may yield biased estimates of the regressor in question, which in turn will yield bias and inconsistency in a second stage regression, discussing (typically theoretical) approaches to resolving the problem. 

Our context and the problem we are seeking to resolve bear obvious similarities to the generated regressor problem. However, the measurement error problem that we address here nonetheless has some unique characteristics that differentiate it. In particular, in our setting, the measurement error stems from predictions of a machine learning model, which is built using a set of labeled data on which the covariate of interest is (assumedly) perfectly observed. In other words, the covariate to be generated is only \textit{partially} unobserved. This is different from the typical setup of a generated regressors model in the literature. This partial observation via a labeled dataset enables objective quantification of measurement error and potentially more effective bias correction. As will be discussed later, our proposed method makes use of this labeled dataset to achieve bias correction. 

The above having been said, some recent work in the generated regressors literature has proposed methods to correct for bias in generated regressors in the presence of distributional information about that bias \citep[e.g.,][]{meng2016linear}. Accordingly, as part of our benchmarking analyses, we seek to compare the relative performance of our method with that approach \citep{meng2016linear}.

\clearpage

\section{Theoretical Setup} \label{FormalTheory}
In this appendix, we provide a formal setup of the measurement error problem as well as the instrumental variable approach to address it. We begin by providing a standard setup of the econometric estimation problem. We assume that in a certain population, the relationship of interest is captured by
\begin{equation}
    \label{eq:pop_reg}
    \boldsymbol{Y} = \boldsymbol{X}\beta_X + \boldsymbol{Z}\beta_Z + \varepsilon,
\end{equation}
where $\varepsilon$ is the random error term.
We observe $\{(y_i,x_i,z_i)_{i=1,\ldots,n_1}\}$, a sample of $n_1$ independent and identically distributed units from the population of interest, such that
\begin{align}
    \begin{split}
        \label{eq:exog_reg}
        y_i&=x_i\beta_X + z_i\beta_Z + \varepsilon_i\quad\text{for } i=1,\ldots, n_1\\
        \boldsymbol{Y}&=\boldsymbol{X}\beta_X + \boldsymbol{Z} \beta_Z + \boldsymbol{\varepsilon} \quad \boldsymbol{Z} \in \mathbb{R}^{n_1\times k} \text{ and } \boldsymbol{Y}, \boldsymbol{\varepsilon}, \boldsymbol{X} \in \mathbb{R}^{n_1\times 1}
   \end{split}
\end{align}
Moreover, we will let $\beta=\left[\beta_X,\beta_Z\right]$, $\boldsymbol{A} = \left[\boldsymbol{X},\boldsymbol{Z}\right]$, and for simplicity assume that the $k+1$ explanatory variables represented by $\boldsymbol{A}$ have zero mean, in addition to the following standard linear regression assumptions:\\[1ex]
\emph{(A1)} $\mathbb{E}\left[\boldsymbol{\varepsilon}|\boldsymbol{A}\right]=0$,\\ 
\emph{(A2)} $rank(\mathbb{E}\left[\boldsymbol{A'}\boldsymbol{A}\right])=k+1$.\\[1ex]
With (A1) and (A2), we have that the Ordinary Least Squares (OLS) estimator $\hat{\beta}_{OLS}=(\boldsymbol{A'}\boldsymbol{A})^{-1}\boldsymbol{A'}\boldsymbol{Y}$ is unbiased and consistent with respect to $\beta$. 

In our context, $\hat{\beta}_{OLS}$ can be estimated using the labeled data (i.e., $n_1 \equiv |D_{label}|$). In addition, we are able to observe another sample $\{(y_i,\widehat{x_i},z_i)_{i=n_1+1,\ldots,n_1+n_2}\}$ of $n_2$ independent and identically distributed units from the population of interest, where $\widehat{x_i}$ (i.e., vector $\boldsymbol{\widehat{X}}$) is an imperfect measurement of $x_i$ (i.e., vector $\boldsymbol{X}$). In our context, $\boldsymbol{X}$ represents the true values of the data-mined variable, and $\boldsymbol{\widehat{X}}$ represents imperfect predictions generated by the machine learning model. For $n_1$ samples, we obtain the true value of $\boldsymbol{X}$ (i.e., the ground truth labels), but it is prohibitive (e.g., in cost or time) to obtain labels for the remaining $n_2$ samples (i.e., $n_2 \equiv |D_{unlabel}|$). Typically, $n_2 >> n_1$. Given its large size, there is a clear desire to utilize the information contained in the $n_2$ samples for inference. We make the following additional assumptions about the imperfect measurement, $\boldsymbol{\widehat{X}}$:\\[1ex]
\emph{(A3)} $\boldsymbol{\widehat{X}} = \boldsymbol{X} + \boldsymbol{e} $,\\
\emph{(A4)} $\mathbb{E}\left[\boldsymbol{\varepsilon'}\boldsymbol{e}\right]=\boldsymbol{0}$,\\
\emph{(A5)} $\mathbb{E}\left[\boldsymbol{X'}\boldsymbol{e}\right]=\boldsymbol{0}$, and $\mathbb{E}\left[\boldsymbol{Z'}\boldsymbol{e}\right]=\boldsymbol{0}$\\[1ex]
Attempting to estimate \eqref{eq:exog_reg} simply by replacing $\boldsymbol{X}$ with $\boldsymbol{\widehat{X}}$ is known to result in a biased and inconsistent estimate of $\hat{\beta}_{OLS}$ because $\boldsymbol{\widehat{X}}$ is endogenous:
\begin{align}
    \begin{split}
        \label{eq:endog_reg}
        y_i&=\widehat{x_i}\beta_X + z_i\beta_Z + \left[\varepsilon_i-e_i\beta_X\right]\quad\text{for } i=n_1+1,\ldots, n_1+n_2,\\
        \boldsymbol{Y}&=\boldsymbol{\widehat{X}}\beta_X + \boldsymbol{Z} \beta_Z + \left[\boldsymbol{\varepsilon}- \boldsymbol{e}\beta_X \right]\quad \boldsymbol{Z} \in \mathbb{R}^{n_2\times k} \text{ and } \boldsymbol{Y}, \boldsymbol{\varepsilon},\boldsymbol{e} ,\boldsymbol{\widehat{X}} \in \mathbb{R}^{n_2\times 1} ,
   \end{split}
\end{align}
where endogeneity derives from the fact that $\mathbb{E}\left(\boldsymbol{\widehat{X}}'\boldsymbol{e}\right)=Var\left(\boldsymbol{e}\right)$. Note that $\left[\cdot\right]$ in the regression equations above and below is notation used (for emphasis) to represent the (unobserved) error term.

Instrumental variable regression is a common approach to resolve issues of endogeneity. It begins with the assumption that we actually observe $(y_i,\widehat{x_i},z_i, w_i)_{i=n_1+1,\ldots,n_1+n_2}$, where $w_i$ is a $d$-dimensional row-vector of (presumed) instrumental variables, with $d\ge 1$. Estimation on $n_2$ samples is carried out in a two-stage least squares (2SLS) regression, of the form: 
\begin{align}
    \begin{split}
        \label{eq:tsls_reg}
        \boldsymbol{\widehat{X}}&=\boldsymbol{W}\boldsymbol{\Lambda_W} + \boldsymbol{Z} \boldsymbol{\Lambda_Z} + \boldsymbol{u}\quad \boldsymbol{W} \in \mathbb{R}^{n_2\times d},\\
        \boldsymbol{Y}&=\widetilde{\boldsymbol{X}}\beta_X + \boldsymbol{Z} \beta_Z + \overbrace{\left[\boldsymbol{\varepsilon}+ \beta_X(\widetilde{\boldsymbol{u}}-\boldsymbol{e}) \right]}^{\boldsymbol{r}},
        \end{split}
\end{align}
where $\widetilde{\boldsymbol{X}} = \boldsymbol{H}_W\boldsymbol{\widehat{X}}$ represents the projection of $\boldsymbol{\widehat{X}}$ onto the column space of $\boldsymbol{B}$, where $\boldsymbol{B} = \left[\boldsymbol{W},\boldsymbol{Z}\right]$ and $\Lambda_B=\left[\Lambda_W,\Lambda_Z\right]$, $\boldsymbol{H}_W= \boldsymbol{B(B'B)^{-1}B'}$; and where $\widetilde{\boldsymbol{u}}=\boldsymbol{\widehat{X}}-\widetilde{\boldsymbol{X}}$. Denote $\boldsymbol{C} = \left[\widetilde{\boldsymbol{X}},\boldsymbol{Z}\right]$, then the 2SLS estimator $\hat{\beta}_{2SLS}=(\boldsymbol{C}'\boldsymbol{C})^{-1}\boldsymbol{C}'\boldsymbol{Y}$, equates to
\begin{align}
    \begin{split}
        \label{eq:tsls_est}
        \hat{\beta}_{2SLS}=&\beta + \left(\frac{\boldsymbol{C}'\boldsymbol{C}}{n_2}\right)^{-1} \left(\frac{\boldsymbol{C}'r}{n_2}\right)\\
        \pconv&\beta + \plim_{n_2\to\infty} \left(\frac{\boldsymbol{C}'\boldsymbol{C}}{n_2}\right)^{-1} \plim_{n_2\to\infty}\left(\frac{\boldsymbol{C}'r}{n_2}\right) = \tilde{\beta}
    \end{split}
\end{align}

Therefore, $\hat{\beta}_{2SLS}$ is a consistent estimator of $\beta$ (i.e., $\tilde{\beta} = \beta$)  under additional standard instrumental variable assumptions:\\
\emph{(A6)} $\mathbb{E}\left[\boldsymbol{B'}\boldsymbol{\varepsilon}\right]=\boldsymbol{0}$,\\ 
\emph{(A7)} $\mathbb{E}\left[\boldsymbol{B'}\boldsymbol{e}\right]=\boldsymbol{0}$,\\ 
\emph{(A8)} $rank(\mathbb{E}\left[\boldsymbol{B'B}\right])=d+k$,\\
\emph{(A9)} $rank(\mathbb{E}\left[\boldsymbol{B'A}\boldsymbol{}\right])=k+1$.

The following theorem formally establishes that, given an endogenous, mis-measured covariate (in this case, a vector of predictions from a single tree within the random forest), as well as a set of \textit{other} mis-measured covariates (in this case, vectors of predictions obtained from other trees comprising the same random forest), in the absence of correlation between the error vector (actual $-$ prediction) associated with the endogenous covariate and the set of other mis-measured covariates,  one can obtain consistent estimates of the coefficients of interest.

\textbf{THEOREM}. \textit{Let the matrix $\boldsymbol{P}=\left[\boldsymbol{p_1}, \ldots, \boldsymbol{p_M}\right] = \boldsymbol{X} + \left[\boldsymbol{e_1}, \ldots, \boldsymbol{e_M}\right]$, where $\forall j \in \{1,\ldots,M\}~\boldsymbol{p_j} \in \mathbb{R}^{n_2\times 1}$ is a column vector measure (with error as defined in \emph{(A3)} - \emph{(A5)}) of variable $\boldsymbol{X}$ under the specification defined in \eqref{eq:endog_reg} with corresponding assumptions \emph{(A1)}-\emph{(A2)}. Additionally, let $\boldsymbol{E}= \boldsymbol{P}-\boldsymbol{X} = \left[\boldsymbol{e_1}, \dots,\boldsymbol{e_M}\right]$ be the matrix of measurement errors,  $S_j\subseteq\{1, \ldots, M\} \setminus j$ be a subset of cardinality $d$ such that $\boldsymbol{P_{S_j}}, \boldsymbol{E_{S_j}} \in \mathbb{R}^{n_2\times d}$ are subsets of the column vectors of $\boldsymbol{P}$ and $\boldsymbol{E}$ respectively defined by the column indices in $S_j$. If $\mathbb{E}\left[\boldsymbol{E_{S_j}'e_j}\right]=\boldsymbol{0}$ then using $P_{S_j}$ as instruments for $P_j$ in 2SLS, provides consistent estimates of the population parameters $\beta$ defined in~\eqref{eq:pop_reg}.}

\begin{proof}
From \eqref{eq:endog_reg} we therefore know that $\boldsymbol{p_j}$ is endogenous because its measurement error $\boldsymbol{e_j}$ is captured by the unobserved error term. While from \eqref{eq:tsls_reg} and \eqref{eq:tsls_est} we know that given a matrix $\boldsymbol{W}$ of instrumental variables, $\hat{\beta}_{2SLS}\pconv\boldsymbol{\beta}$ when \emph{(A6)}-\emph{(A9)} are satisfied. Therefore, it suffices to show that \emph{(A6)}-\emph{(A9)} are satisfied when we let $\boldsymbol{W}=\boldsymbol{P_{S_j}}$.
\begin{align*}
        \emph{(A1)}, \emph{(A3)}-\emph{(A5)} &\implies \mathbb{E}\left[\boldsymbol{Z'}\boldsymbol{\varepsilon}\right]=\boldsymbol{0},~\mathbb{E}\left[\boldsymbol{P_{S_j}'}\boldsymbol{\varepsilon}\right]=\boldsymbol{0}\\
        &\implies \mathbb{E}\left[\boldsymbol{B'}\boldsymbol{\varepsilon}\right]=\boldsymbol{0} \emph{(A6)}, \\
        \\
        \mathbb{E}\left[\boldsymbol{E_{S_j}'}\boldsymbol{e_j}\right]=\boldsymbol{0}&\implies \mathbb{E}\left[\boldsymbol{P_{S_j}'}\boldsymbol{e_j}\right]=\boldsymbol{0}\\
        &\implies \mathbb{E}\left[\boldsymbol{\left[\boldsymbol{X_{S_j}}+\boldsymbol{E_{S_j}},Z\right]'}\boldsymbol{\varepsilon}\right]=\boldsymbol{0}\\
        &\implies \mathbb{E}\left[\boldsymbol{B'}\boldsymbol{\varepsilon}\right]=\boldsymbol{0} \emph{(A7)},
\end{align*}
and finally \emph{(A8)}-\emph{(A9)} follow directly from \emph{(A1)}-\emph{(A5)}, recognizing that $\mathbb{E}\left[\frac{\boldsymbol{p_j'}\boldsymbol{p_l}}{n_2}\right] = Var[X]~j,l \in \{1,\ldots,M\}$.
\end{proof}

\clearpage

\section{Pseudocode for Instrumental Variables Selection Procedure} \label{Algorithm1}

\begin{algorithm}
\caption{Instrumental Variables Selection Procedure}
\label{algorithm:IVSelect}
\KwData{Individual trees' predictions $\boldsymbol{P} = \{\widehat{X}^{(1)}, \dots ,\widehat{X}^{(M)}\}$ on $D_{test}$ and $D_{unlabel}$ and ground truth $X$ on $D_{test}$}
\textbf{Notation}: denote $\left \Vert .\right \Vert_1$ as the L1-norm, $\left \Vert .\right \Vert_2$ as the L2-norm, and $\lambda$ as the lasso penalty level\;
Set $\widehat{X}^{(i)}$ as the endogenous covariate\;
Set $\boldsymbol{P_{-i}} \gets \boldsymbol{P} \setminus \widehat{X}^{(i)}$ as the pool of candidate instruments\;
Set $CurrIVs \gets \boldsymbol{P_{-i}}$\;
\While{True}{
    \tcp{Step 1: removal of invalid instruments}
    Obtain $e^{(i)} = \widehat{X}^{(i)} - X$ on $D_{test}$\;
    Estimate lasso regression ${\displaystyle \min_{\boldsymbol{\Delta}}}\left\{ \dfrac{1}{|D_{test}|} \Vert e^{(i)} - \sum_{\widehat{X}^{(j)} \in CurrIVs} \delta_j \widehat{X}^{(j)} \Vert_{2}^{2}+\lambda\Vert\boldsymbol{\Delta}\Vert_{1}\right\}$\;
    Get $\boldsymbol{V_i} \gets \{\widehat{X}^{(j)} \in CurrIVs | \delta_j=0 \}$ as the set of instruments with \textit{zero} coefficients\;
    \tcp{Step 2: selection of strong instruments}
    On $D_{test} \cup D_{unlabel}$, estimate lasso regression ${\displaystyle \min_{\boldsymbol{\Gamma}}}\left\{ \dfrac{1}{|D_{test} \cup D_{unlabel}|} \Vert \widehat{X}^{(i)} - \sum_{\widehat{X}^{(j)} \in \boldsymbol{V_i}} \gamma_j \widehat{X}^{(j)} \Vert_{2}^{2}+\lambda\Vert\boldsymbol{\Gamma}\Vert_{1}\right\}$\;
    Get $\boldsymbol{S_i} \gets \{\widehat{X}^{(j)} \in \boldsymbol{V_i} | \gamma_j \neq 0 \}$ as the set of instruments with \textit{non-zero} coefficients\;
    \If{$\boldsymbol{S_i} == CurrIVs$}{
        \textbf{Break}; \tcp{remaining instruments are valid and strong}
    }{
        Set $CurrIVs \gets \boldsymbol{S_i}$; \tcp{repeat selection}
    }
}
\KwOut{$\boldsymbol{S_i}$, the set of valid and strong instruments for $\widehat{X}^{(i)}$.}
\end{algorithm}

\clearpage

\section{Pseudocode for ForestIV Approach} \label{Algorithm2}

\begin{algorithm}
\caption{Pseudocode for ForestIV Approach}
\label{algorithm:ForestIV}
\KwData{Individual trees' predictions $\boldsymbol{P} = \{\widehat{X}^{(1)}, \dots ,\widehat{X}^{(M)}\}$ on $D_{test}$ and $D_{unlabel}$ and ground truth $X$ on $D_{test}$}
Estimate $\widehat{\boldsymbol{\beta}}_{label}$ as unbiased coefficients on $D_{label}$\;
\ForEach{$i \in \{1, \dots, M\}$}{
    Set $\widehat{X}^{(i)}$ as the endogenous covariate in the econometric model of interest\;
    Select $\boldsymbol{S_i} \subseteq \boldsymbol{P} \setminus \widehat{X}^{(i)}$ using Algorithm \ref{algorithm:IVSelect}\;
    \tcp{Step 3: estimation}
    \If{$\boldsymbol{S_i} \neq \emptyset$}{
        Use $\boldsymbol{S_i}$ as instrumental variables, estimate $\widehat{\boldsymbol{\beta}}_{IV}^i$ on $D_{unlabel}$, with variance-covariance matrix $\widehat{\boldsymbol{\Sigma}}_{IV}^i$\;
        Calculate Hotelling's $T^2$ statistic, $H_i$, between $\widehat{\boldsymbol{\beta}}_{IV}^i$ and $\widehat{\boldsymbol{\beta}}_{label}$\;
        \If{$H_i < Critical Value$}{
            Retain $\widehat{\boldsymbol{\beta}}_{IV}^i$\;
            Calculate $MSE_i = tr\left((\widehat{\boldsymbol{\beta}}_{IV}^i - \widehat{\boldsymbol{\beta}}_{label}) (\widehat{\boldsymbol{\beta}}_{IV}^i - \widehat{\boldsymbol{\beta}}_{label})^T + \widehat{\boldsymbol{\Sigma}}_{IV}^i\right)$\;
        }
    }
}
\KwOut{The retained $\widehat{\boldsymbol{\beta}}_{IV}^i$ with the smallest $MSE_i$. Use bootstrapping to obtain estimation of variances.}
\end{algorithm}

\clearpage

\section{Simulation Results for Continuous Endogenous Covariate with Bike Sharing Data} \label{BikeResults}

\subsection{Effectiveness of Two-Step Selection of Instrumental Variables} \label{Bike_Selection}
On the Bike Sharing data, we calculate two metrics using $D_{test}$, where both the ground truth $lnCnt$ and the random forest's predictions are observed. Before the two-step lasso-based selection procedure, for a given individual tree $i$, we treat \textit{all} other trees' predictions as instruments. We calculate (1) the $F$-statistic associated with a 2SLS regression, which serves as an illustrative measure of instrument strength, and (2) the adjusted $R^2$ associated with an OLS regression of tree $i$'s prediction error on the instruments, which serves as an illustrative measure of instrument exclusion. Observing small $F$-statistic and large $R^2$ indicate having weak and invalid instruments. After the selection, we again calculate these two metrics, using only \textit{selected} instruments. We plot the $F$-statistics against the adjusted $R^2$, both before and after the lasso-based selections, in the following Figure \ref{fig:scatter}.

\begin{figure}[tbhp]
\centering
\label{fig:scatter}
\includegraphics[width=0.9\linewidth]{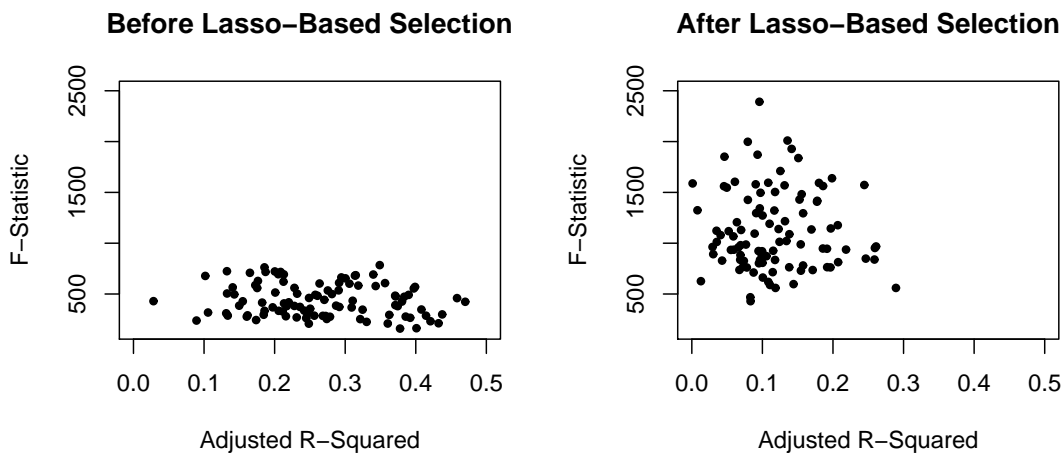}
\caption{Plot of $F$-Statistic against Adjusted $R^2$ based on a Single Simulation Run. The left plot corresponds to 2SLS estimations using all candidate instruments (i.e., without selection), and the right plot corresponds to 2SLS estimations using only the selected instruments.}
\end{figure}

We can see that, after the lasso-based selections, the $F$-statistics become much larger ($p < 0.001$), and the adjusted $R^2$ become even smaller ($p < 0.001$). This provides descriptive evidence that our lasso-based procedures are useful in further selecting strong and valid instruments.

\subsection{Main Results} \label{Bike_MainResults}

In each simulation round, we record the estimates and standard errors associated with each coefficient. In Table \ref{table:BikeData}, we report the average coefficients and standard errors (in parentheses) across all simulation rounds, both for the biased regression and the unbiased regression (obtained on the labeled data). For ForestIV, we report the average coefficients and the standard deviations of the sampling distributions across 100 simulation runs as standard errors. We also compute the $p$-value (in square brackets) associated with a $t$-test \textit{comparing each estimated coefficient with its underlying true value} (i.e., a larger $p$-value means an estimate is statistically closer to the true value). We can see that, directly using the $\widehat{lnCnt}$ predicted by random forest in the regression model results, on average, in 13.2\% overestimation on $lnCnt$ and 29.8\% underestimation on the intercept. ForestIV estimates are not significantly different from the true $lnCnt$ values, as indicated by the $p$-values. Compared with the unbiased estimates, ForestIV achieves smaller standard errors (i.e., higher estimation precision).

\begin{table}[tbhp]
\label{table:BikeData}
    \centering
    \begin{tabular}{c c c c c}
    \hline
         & True & Biased & Unbiased & ForestIV \\ 
         \hline
        Intercept       & 1.0 & 0.702 (0.063)   & 1.018 (0.204) &  0.957 (0.134) \\
                        &     & [0.004]     & [0.511]   & [0.745]   \\
        $lnCnt$         & 0.5 &  0.566 (0.013)  & 0.498 (0.040) &  0.512 (0.027) \\
                        &     & [0.002]     & [0.530]   & [0.652]   \\
        $Z_1$           & 2.0 & 2.000 (0.003)   & 1.999 (0.011) &  2.000 (0.003) \\
                        &     & [0.459]     & [0.524]   & [0.977]   \\
        $Z_2$           & 1.0 & 1.000 (0.002)   & 0.999 (0.006) &  1.000 (0.002) \\
                        &     & [0.480]     & [0.486]   & [0.989]   \\
        \hline
        Ave\_MSE        &     & 0.150 & 0 & 0.017 \\
    \hline
    \end{tabular}
    \caption{ForestIV Results on Bike Sharing Data. Standard errors in parentheses. $p$-values comparing estimates with true values in square brackets. Ave\_MSE contains the average empirical MSE associated with each set of estimates across 100 simulation runs.}
\end{table}

\subsection{Sensitivity Analyses} \label{Bike_Sensitivity}

Below we present the sensitivity analyses results on Bike Sharing data.

\subsubsection{Size of Unlabeled Data}\label{Bike_Sensitivity_Size}

We repeat the main simulation with 8 different sizes of unlabeled data, respectively 100, 500, 1,000, 5,000, 7,500, 10,000, 12,500, and 16,179 (i.e., all remaining instances). Other parameters, e.g., the size of the labeled data set and the econometric model specifications, are kept unchanged across these simulations. This sensitivity analysis examines how the ForestIV estimate for $lnCnt$ converges as the 2SLS estimations are exposed to larger volumes of unlabeled data. The results are plotted in Figure \ref{fig:UnlabelSize}. The 95\% confidence interval of the ForestIV estimate is constructed as the range between the 2.5th percentile and 97.5th percentile of the estimate's sampling distribution.

\begin{figure}[tbhp]
\label{fig:UnlabelSize}
\centering
\includegraphics[width=0.8\linewidth]{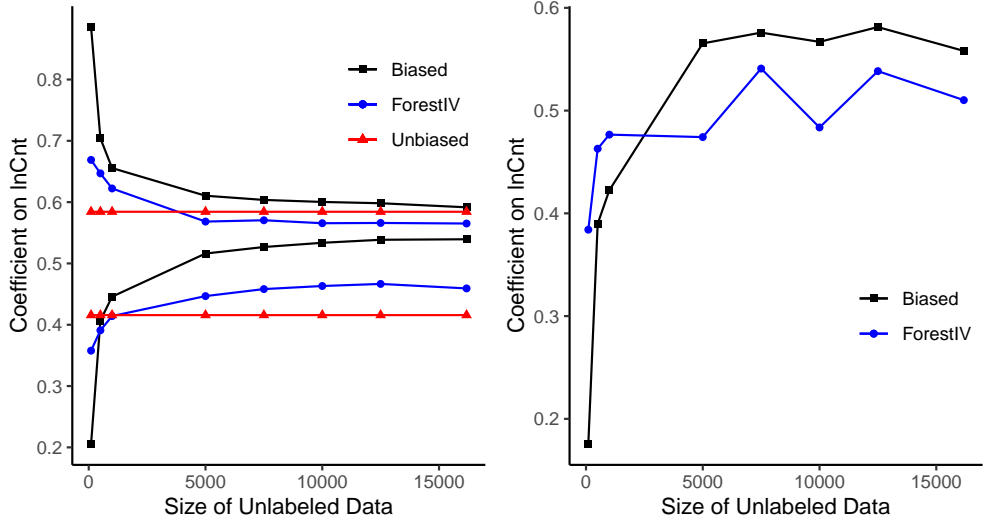}
\caption{Confidence Interval and Point Estimate on $lnCnt$ with Different Sizes of Unlabeled Data. The left panel shows the 95\% confidence interval of the estimation on $lnCnt$ with different sizes of unlabeled data. The right panel shows the point estimate on $lnCnt$ with different sizes of unlabeled data based on single simulation run. The true coefficient on $lnCnt$ is 0.5.}
\end{figure}

We observe that the biased estimate stays biased regardless of increases in the volume of unlabeled data, which is expected given the fact that measurement error results in inconsistent estimates. On the other hand, the confidence interval of ForestIV estimate becomes narrower as more unlabeled data is added. In other words, with sufficient unlabeled data, ForestIV estimates converge and achieve more precision than the unbiased estimates. The point estimates from a single simulation run, despite some local fluctuations, support the same observations.

This sensitivity analysis provides additional support for the quality of instruments generated from the random forest, and subsequently selected by our proposed procedure. In practice, the benefit of combining machine learning with econometric modeling is perhaps most salient when the unlabeled data is much larger than the labeled data, particularly when acquiring a large amount of labeled data is prohibitively expensive. ForestIV can offer substantial utility in such a scenario, as its estimates converge as the volume of unlabeled data grows larger, and this ultimately brings gains in precision relative to the unbiased estimates obtained using the (relatively small) labeled dataset.

\subsubsection{Noise in Econometric Data}\label{Bike_Sensitivity_Noise}
We repeat the simulation by increasing $\sigma_{\varepsilon}$ from 2 to 5 (or equivalently, increasing the error term variance by 6.25 times). This represents the scenario where the data for econometric estimation contains a large amount of noise. The results are reported in Table \ref{table:BikeNoise}.

\begin{table}[tbhp]
\label{table:BikeNoise}
    \centering
    \begin{tabular}{c c c c c}
    \hline
         & True & Biased & Unbiased & ForestIV \\ 
         \hline
        Intercept       & 1.0 & 0.692 (0.155)  & 1.044 (0.512)    &  0.870 (0.270) \\
                        &     & [0.157]   & [0.511]   & [0.630]   \\
        $lnCnt$         & 0.5 &  0.568 (0.033)   & 0.495 (0.089)  &  0.531 (0.055) \\
                        &     & [0.143]   & [0.530]   & [0.576]   \\
        $Z_1$           & 2.0 & 2.000 (0.007)    & 1.998 (0.027)  &  2.000 (0.008) \\
                        &     & [0.464]   & [0.524]   & [0.982]   \\
        $Z_2$           & 1.0 & 1.000 (0.004)    & 0.998 (0.016)  &  1.000 (0.004) \\
                        &     & [0.488]   & [0.486]   & [0.998]   \\
        \hline
        Ave\_MSE        &     & 0.408 & 0 & 0.143 \\
    \hline
    \end{tabular}
    \caption{ForestIV Results on Bike Sharing Data with $\sigma_{\varepsilon}=5$. Standard errors in parentheses. $p$-values comparing estimates with true values in square brackets. Ave\_MSE contains the average empirical MSE associated with each set of estimates across 100 simulation runs.}
\end{table}

We observe that the average standard errors associated with biased, unbiased, and ForestIV estimates all become larger (roughly by a factor of 2.5) compared to the case when $\sigma_{\varepsilon}=2$, which is a direct consequence of increased noise in the data. However, despite greater noise in the data, ForestIV still achieves significant gains in the precision of estimates, relative to the unbiased estimates, as the ForestIV estimates yield substantially smaller standard errors.

\subsubsection{Predictive Performance of Random Forest} \label{Bike_Sensitivity_NTree}
Thus far we have treated the random forest model as fixed. However, the standard procedure of building a predictive machine learning model typically involves tuning various parameters to achieve the best predictive performance on the hold-out labeled data. Therefore, in the next set of sensitivity analyses, we examine the relationship between the predictive performance of random forest (measured using Root Mean Square Error, RMSE, a common performance metric in machine learning) and the correction performance of ForestIV. Given that many parameters of a random forest can be fine-tuned, we choose to focus on the total number of trees (i.e., $M$), because this parameter is directly related to the number of candidate instruments. 

We repeat the simulation with 5 different values of $M$: 25, 50, 100, 150, and 200. Across these simulations, the other parameters remain as described in the main simulation. We report ForestIV estimation results for each choice of $M$ in Table \ref{table:BikeNTree}. We also report the average Hotelling $T^2$-statistic that compares $\widehat{\boldsymbol{\beta}}_{label}$ with ForestIV estimates in the last row of the table. In Figure \ref{fig:BikeNtree}, we plot the distribution of ForestIV estimation on $lnCnt$ for three particular sizes of random forest, i.e., $M \in \{25, 100, 200\}$.

\begin{table}[tbhp]
\label{table:BikeNTree}
    \centering
    \small
    \begin{tabular}{c c c c c c c}
    \hline
        &  & $M=25$ & $M=50$ & $M=100$ & $M=150$ & $M=200$ \\ 
        \hline
        RMSE  &  & 0.6657 & 0.6552 & 0.6436 & 0.6430 & 0.6411\\
        \hline
        & True & & & & \\
        \hline
            Intercept & 1.0 & 0.884 (0.132)  & 0.919 (0.120)  & 0.957 (0.133) & 0.941 (0.133) & 0.889 (0.132)    \\
                        &  & [0.381]    & [0.500]   & [0.745]  & [0.656]  & [0.519]     \\
            $lnCnt$   & 0.5 & 0.526 (0.027)  & 0.519 (0.025)  & 0.512 (0.027) & 0.514 (0.027) & 0.524 (0.027)    \\
                        &  & [0.339]    & [0.440]   & [0.652]  & [0.590]  & [0.482]     \\
            $Z_1$     & 2.0 & 2.000 (0.003)  & 2.000 (0.003)  & 2.000 (0.003) & 2.000 (0.003) & 2.000 (0.003)    \\
                        &  & [0.892]    & [0.920]   & [0.977]  & [0.911]  & [0.777]     \\
            $Z_2$     & 1.0 & 1.000 (0.002)  & 1.000 (0.002)  & 1.000 (0.002) & 1.000 (0.002) & 1.000 (0.002)    \\
                        &  & [0.916]    & [0.924]   & [0.989]  & [0.947]  & [0.897]     \\
            \hline
            Ave\_MSE  &  & 0.016 & 0.020 & 0.017 & 0.016 & 0.017 \\
        \hline
        Ave. Hotelling $T^2$ stats & & 3.0817 & 3.0501 & 2.7094 & 2.7816 & 2.9357 \\
        \hline
    \end{tabular}
    \caption{ForestIV Estimates for Different Choices of $M$. Standard errors in parentheses. $p$-values comparing estimates with true values in square brackets. Ave\_MSE contains the average empirical MSE associated with each set of estimates across 100 simulation runs.}
\end{table}

\begin{figure}[tbhp]
\label{fig:BikeNtree}
\centering
\includegraphics[width=0.47\linewidth]{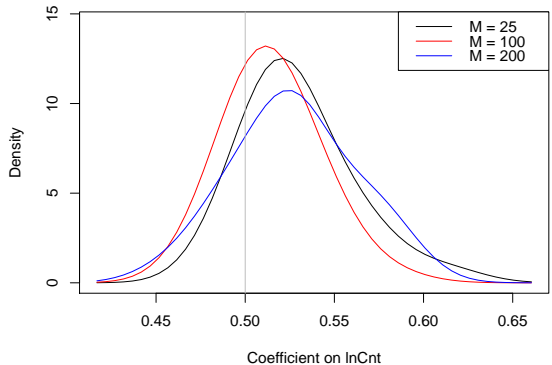}
\caption{ForestIV Estimation on $lnCnt$ for Different Choices of $M$}
\end{figure}

Taking the average point estimates for each $M$ as indications of correction performance, we see that having too few trees (e.g., $M=25$ or $M=50$) tends to result in higher RMSE \textit{and} poor ForestIV results. Meanwhile, as the number of trees increases, we observe marginal improvements on RMSE, but this does not always lead to better correction performance. Specifically, when $M=150$ and $M=200$, we observe slight improvements in RMSE scores compared to when $M=100$, but the correction performance of ForestIV declines. This is also supported by a visual examination of \ref{fig:BikeNtree}, where we observe better correction performance for $M=100$ than for $M=25$ or $M=200$. Importantly, we calculate the average Hotelling $T^2$ statistic that compares $\widehat{\boldsymbol{\beta}}_{label}$ with ForestIV estimates, and find it to align with ForestIV's correction performance. Smaller $T^2$ statistic indicates better correction performance. 

This set of simulations implies that ForestIV cannot replace efforts to tune the hyper-parameters associated with the random forest prediction model. That is, having a better-performing random forest tends to improve the performance of ForestIV on average. At the same time, in this particular demonstration, we observe that using an excessive number of trees can potentially hurt ForestIV's correction performance, likely because instrument selection can be challenging with a large number of candidate instruments. In practice, researchers can rely on the Hotelling $T^2$ statistic as a signal, and select the number of trees that minimizes $T^2$ statistic.

\subsection{Simulation Results on Alternative Designs of ForestIV}\label{Bike_AlternativeDesign}

\subsubsection{Sample Splitting}\label{Bike_AlternativeDesign_SampleSplit}
An alternative way of constructing candidate instrumental variables is sample splitting, which is a general strategy in econometrics to leverage multiple independent samples for estimation and inference \citep[e.g.,][]{angrist1995split,chernozhukov2016double,athey2017state}. Consider splitting the training data into two independent subsets, and build one random forest model on each subset. Each forest is asymptotically consistent (as the size of training data goes to infinity), and its prediction error is characterized as a random and independent noise term \citep{scornet2015consistency}. Therefore, predictions from one random forest could, in principle, serve as an instrument for predictions from the other. The potential appeal of this sample splitting approach is that it (to some degree) preserves the superior predictive performance of a random forest over an individual decision tree, and thus may reduce the extent of the measurement error problem up front.

We explore the correction performance of this alternative approach with a simulation experiment on the Bike Sharing data. The basic setup is the same as in Section 3.1 of the main manuscript, except that we randomly split the training data into two equally-sized samples (each with 500 instances), and train one random forest model with 100 trees on each sample. We refer to the two forests as Forest\#1 and Forest\#2. We then use Forest\#1's predictions on the unlabeled data as the endogenous covariate, and use Forest\#2's predictions on the unlabeled data as the instrumental variable. Because there is only one instrument in this case, no instrument selection procedure is needed. We report the average estimates of this sample splitting approach over 100 simulation runs, together with the biased, unbiased, and ForestIV average estimates, in Table \ref{table:BikeData_SampleSplit}.

\begin{table}[tbhp]
\label{table:BikeData_SampleSplit}
    \centering
    \begin{tabular}{c c c c c c}
    \hline
         & True & Biased & Unbiased & ForestIV & Sample Splitting \\ 
         \hline
        Intercept       & 1.0 & 0.702 (0.063)   & 1.018 (0.204) &  0.957 (0.134) & 0.553 (0.115)\\
                        &     & [0.004]     & [0.511]   & [0.745]   & [0.000]\\
        $lnCnt$         & 0.5 &  0.566 (0.013)  & 0.498 (0.040) &  0.512 (0.027) & 0.599 (0.023)\\
                        &     & [0.002]     & [0.530]   & [0.652]  & [0.000]\\
        $Z_1$           & 2.0 & 2.000 (0.003)   & 1.999 (0.011) &  2.000 (0.003) & 2.000 (0.003)\\
                        &     & [0.459]     & [0.524]   & [0.977]  & [0.961]\\
        $Z_2$           & 1.0 & 1.000 (0.002)   & 0.999 (0.006) &  1.000 (0.002) & 1.000 (0.001)\\
                        &     & [0.480]     & [0.486]   & [0.989]  & [0.977]\\
        \hline
        Ave\_MSE        &     & 0.150 & 0 & 0.017 & 0.283\\
    \hline
    \end{tabular}
    \caption{Sample Splitting: Results on Bike Sharing Data. Standard errors in parentheses. $p$-values comparing estimates with true values in square brackets. Ave\_MSE contains the average empirical MSE associated with each set of estimates across 100 simulation runs.}
\end{table}

We find that the estimates produced by the sample splitting approach are in fact even more biased than the biased estimates (taken directly from the machine learning model without any correction), which is a strong indication that the instrument is invalid \citep{murray2006avoiding}. Indeed, after examining the predictions from both random forests on the testing data, we find that, while Forest\#1's predictions are strongly correlated with Forest\#2's predictions (average correlation is 0.96, indicating strong instrument relevance), Forest\#1's prediction \textit{errors} are \textit{not} weakly correlated with Forest\#2's predictions (average correlation is 0.30, indicating systematic violations of the instrument exclusion requirement). Therefore, although sample splitting may be a theoretically feasible approach to construct instruments, it does not appear to be effective in this case. We further repeat the simulations for a few additional sizes of training data, and summarize the results in Table \ref{table:BikeData_SampleSplitSize}. We again observe that ForestIV consistently outperforms sample splitting. 

\begin{table}[tbhp]
\label{table:BikeData_SampleSplitSize}
    \centering
    \begin{tabular}{c c c c c c c}
    \hline
         & True & $|D_{train}|$ & 250 & 500 & 1000 & 2000 \\ 
         &  & $|D_{train}|/2$ & 125 & 250 & 500 & 1000 \\ 
         \hline
        Intercept       & 1.0 & ForestIV & 0.951 (0.377) & 0.872 (0.242) &  0.957 (0.134) & 0.903 (0.096)\\
                        &  & Sample Splitting & 0.259 (0.283)  & 0.421 (0.180) &  0.553 (0.115) & 0.631 (0.090) \\
        $lnCnt$         & 0.5 & ForestIV & 0.514 (0.079) & 0.528 (0.050) &  0.512 (0.027) & 0.522 (0.021)\\
                        &  & Sample Splitting & 0.664 (0.056)  & 0.627 (0.035) &  0.599 (0.023) & 0.581 (0.018) \\
        $Z_1$           & 2.0 & ForestIV & 1.999 (0.003) & 2.000 (0.003) &  2.000 (0.003) & 2.000 (0.002)\\
                        &  & Sample Splitting &  2.000 (0.003)  & 1.999 (0.002) &  2.000 (0.003) & 2.000 (0.003)\\
        $Z_2$           & 1.0 & ForestIV & 1.000 (0.002) & 1.000 (0.002) &  1.000 (0.002) & 1.000 (0.001) \\
                        &  & Sample Splitting &  1.000 (0.002)  & 1.000 (0.002) &  1.000 (0.002) & 1.000 (0.002)\\
        \hline
        Ave\_MSE        &  & ForestIV & 0.040 & 0.010 & 0.017 & 0.021\\
                        &  & Sample Splitting & 0.664 & 0.391 & 0.283 & 0.156 \\
    \hline
    \end{tabular}
    \caption{Sample Splitting with Different Sizes of Training Data. Standard errors in parentheses. $|D_{train}|$ represents the size of training data, and $|D_{train}|/2$ is the size of each split sample. Ave\_MSE contains the average empirical MSE associated with each set of estimates across 100 simulation runs.}
\end{table}

In addition, we extend the simulations to 13 different sizes of training data: 200, 250, 300, 350, 400, 450, 500, 750, 1000, 1250, 1500, 1750, and 2000. For each size of training data, we apply both ForestIV and sample splitting, then calculate the average squared bias, variance, and MSE associated with each estimator using the true coefficient values (rather than the unbiased estimates) across 100 simulation runs. In Figure \ref{fig:SampleSplit}, we plot the squared bias, variance, and MSE of the two estimators across different sample sizes, both for the coefficient on $lnCnt$ (i.e., the variable generated by machine learning) and for all coefficients in the regression model. We also plot the ratios of squared bias, variance, and MSE of the two estimators. Empirically, we find that ForestIV has uniformly smaller bias, variance, and MSE than sample splitting, and the ratios between the two are always smaller than 1. Moreover, we observe a downward trend in the bias ratio and MSE ratio of ForestIV over sample splitting as the size of training data gets larger. This empirical evidence indicates that ForestIV may be a more efficient estimator than sample splitting. With a fixed size of training data, the former achieves smaller bias and smaller variance (i.e., better estimation) than the latter, and the advantage of ForestIV seems to become even larger as the size of training data increases.\footnote{We acknowledge that, if the size of training data were significantly larger, performance of the sample splitting approach could potentially further improve. However, once again, we must reiterate that, in the presence of more training data, there is no need to ``mine" covariates via machine-learning techniques in the first place.}

\begin{figure}[tbhp]
\label{fig:SampleSplit}
\centering
\includegraphics[width=\linewidth]{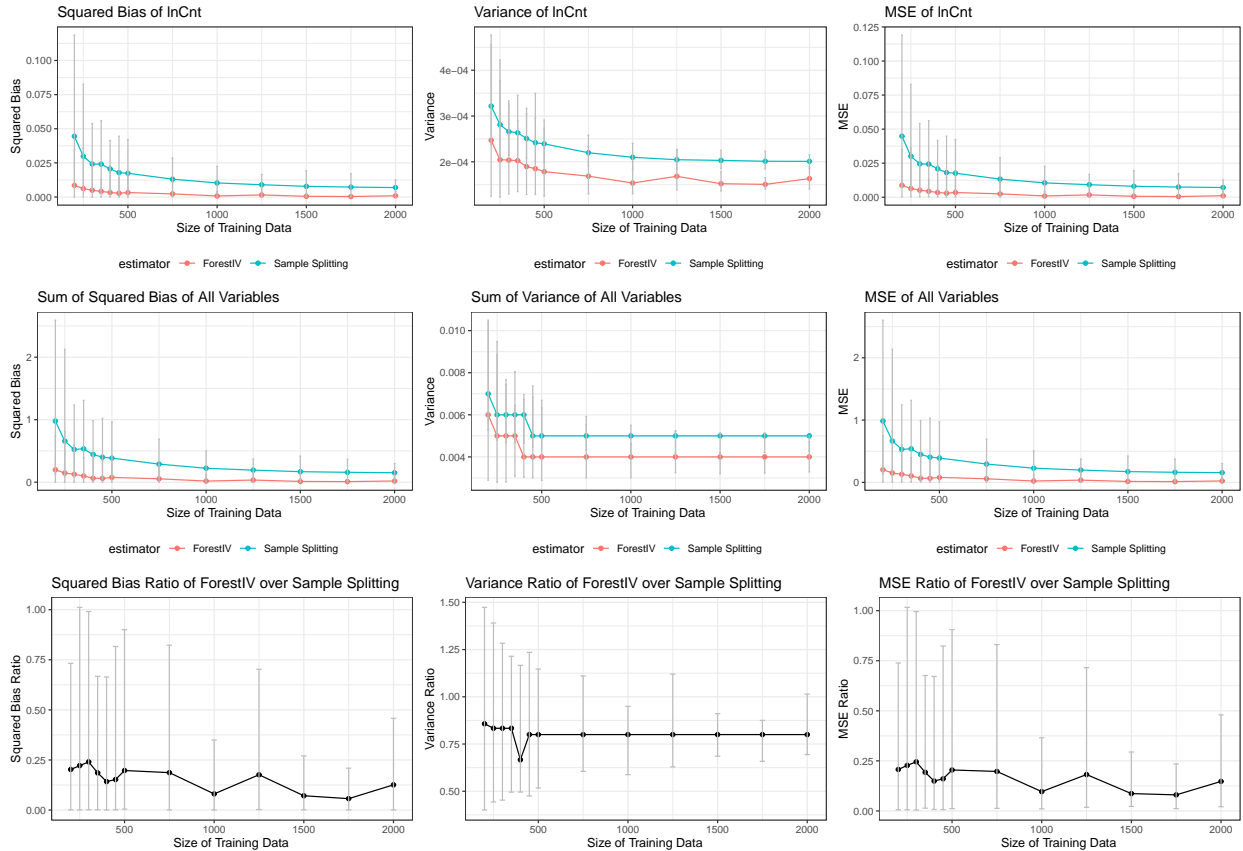}
\caption{Squared Bias, Variance, and MSE Comparisons between ForestIV and Sample Splitting. Vertical bars represent 95\% confidence intervals. The three plots in the first row show the squared bias, variance, and MSE of $lnCnt$, the variable generated by machine learning. The three plots in the second row show the squared bias, variance, and MSE of all variables in the model. The three plots in the third row show the ratios of squared bias, variance, and MSE between the two estimators, when all variables are considered.}
\end{figure}

\subsubsection{Using Subset of Trees as Endogenous Covariate and Instruments} \label{Bike_AlternativeDesign_SubsetTree}
While we propose to use a single tree in the random forest as the endogenous covariate and, under this restriction, select instruments for it, it is certainly possible to use a subset of trees in the forest to construct the endogenous covariate. The potential benefit of this approach is that aggregating over multiple trees can produce more accurate predictions, suggesting an up-front reduction of the measurement error problem. However, this approach poses a computational challenge, because there are $O(2^M)$ possible subsets of trees in a forest comprised of $M$ trees, which are generally infeasible to enumerate and process exhaustively.

Nonetheless, we conduct a preliminary exploration of this strategy. Rather than enumerating all possible subsets, we randomly sample $q\%$ of all trees and use their averaged predictions as the endogenous covariate. Then, from the remaining $(1-q)\%$ of trees, we select the proper instruments (each instrument is still the predictions from a single tree) to estimate the regression. The tree sampling and instrument selection procedure is repeated 100 times, which produces 100 tuples of endogenous and instrumental variables. Note that we only generate 100 such tuples, again because it is computationally burdensome to examine all possible tree subsets even for a fixed $q\%$ (e.g., there are $\binom{100}{30} \approx 3\times 10^{25}$ ways to sample 30 trees from a random forest of 100 trees). As with ForestIV, we conclude by selecting the tuple that minimizes the empirical MSE. We use the same basic simulation configurations under the Bike Sharing dataset, and vary $q\%$ over three levels, 30\%, 50\%, and 70\%, respectively representing low, medium, and high levels of subset aggregation. For each choice of $q\%$, we report the average estimates (along with the biased, unbiased, and ForestIV estimates) across 100 simulation runs, in Table \ref{table:BikeData_SubsetTree}.

\begin{table}[tbhp]
\label{table:BikeData_SubsetTree}
    \centering
    \footnotesize
    \begin{tabular}{c c c c c c c c}
    \hline
         & True & Biased & Unbiased & ForestIV & 30\% Sampling & 50\% Sampling & 70\% Sampling \\ 
         \hline
        Intercept       & 1.0 & 0.702 (0.063)   & 1.018 (0.204) &  0.957 (0.134) & 0.757 (0.083) & 0.718 (0.080) & 0.763 (0.064) \\
                        &     & [0.004]     & [0.511]   & [0.745]   & [0.004] & [0.000] & [0.000]\\
        $lnCnt$         & 0.5 &  0.566 (0.013)  & 0.498 (0.040) &  0.512 (0.027) & 0.552 (0.018) & 0.562 (0.016) & 0.553 (0.014)\\
                        &     & [0.002]     & [0.530]   & [0.652]  & [0.004] & [0.000] & [0.000]\\
        $Z_1$           & 2.0 & 2.000 (0.003)   & 1.999 (0.011) &  2.000 (0.003) & 2.000 (0.003) & 1.999 (0.003) & 1.999 (0.003)\\
                        &     & [0.459]     & [0.524]   & [0.977]  & [0.885] & [0.822] & [0.667]\\
        $Z_2$           & 1.0 & 1.000 (0.002)   & 0.999 (0.006) &  1.000 (0.002) & 1.000 (0.002) & 1.000 (0.002) & 1.000 (0.001)\\
                        &     & [0.480]     & [0.486]   & [0.989]  & [0.846] & [0.947] & [0.944]\\
        \hline
        Ave\_MSE        &     & 0.150 & 0 & 0.017 & 0.096 & 0.166 & 0.132 \\
    \hline
    \end{tabular}
    \caption{Using Subset of Trees: Results on Bike Sharing Data. Standard errors in parentheses. $p$-values comparing estimates with true values in square brackets. Ave\_MSE contains the average empirical MSE associated with each set of estimates across 100 simulation runs.}
\end{table}

We find that randomly selecting a subset of trees to construct the endogenous covariate exhibits limited effectiveness in bias correction, as compared to our proposed ForestIV approach. Across three different sampling ratios, the estimates achieve only small improvements over the biased estimates. Note that we have explored only a highly limited number of subsets, compared to all $O(2^{100})$ possible subsets, and we have restricted each instrument to be the predictions from a single tree (rather than the aggregated predictions from a subset of trees). Fully exploring all possible pairs of endogenous and instrumental variables (each consisting of tree subsets) is clearly infeasible. Future research might further develop this avenue of inquiry, studying potential heuristic or optimization methods to reduce the computational burden of extensive tree subset enumeration. 

\subsubsection{Averaging across Multiple Estimates} \label{Bike_AlternativeDesign_Averaging}
We now examine whether averaging across multiple estimates in the ForestIV procedure may enhance bias correction. Specifically, in Step 3 of ForestIV, rather than selecting a single set of estimates that minimizes empirical MSE, we \textit{average} across all sets of estimates that are not rejected by the Hotelling's $T^2$ test. The averaging estimates, along with the biased, unbiased, and ForestIV estimates, are reported in Tables \ref{table:BikeData_Averaging}.

\begin{table}[tbhp]
\label{table:BikeData_Averaging}
    \centering
    \begin{tabular}{c c c c c c}
    \hline
         & True & Biased & Unbiased & ForestIV & Averaging ForestIV \\ 
         \hline
        Intercept       & 1.0 & 0.702 (0.063)   & 1.018 (0.204) &  0.957 (0.134) & 0.797 (0.086)\\
                        &     & [0.004]     & [0.511]   & [0.745]   & [0.018]\\
        $lnCnt$         & 0.5 &  0.566 (0.013)  & 0.498 (0.040) &  0.512 (0.027) & 0.545 (0.018)\\
                        &     & [0.002]     & [0.530]   & [0.652]  & [0.011]\\
        $Z_1$           & 2.0 & 2.000 (0.003)   & 1.999 (0.011) &  2.000 (0.003) & 2.000 (0.003)\\
                        &     & [0.459]     & [0.524]   & [0.977]  & [0.977]\\
        $Z_2$           & 1.0 & 1.000 (0.002)   & 0.999 (0.006) &  1.000 (0.002) & 1.000 (0.002)\\
                        &     & [0.480]     & [0.486]   & [0.989]  & [0.987]\\
        \hline
        Ave\_MSE        &     & 0.150 & 0 & 0.017 & 0.101\\
    \hline
    \end{tabular}
    \caption{Averaging Estimates: Results on Bike Sharing Data. Standard errors in parentheses. $p$-values comparing estimates with true values in square brackets. Ave\_MSE contains the average empirical MSE associated with each set of estimates across 100 simulation runs.}
\end{table}

\begin{figure}[tbhp]
\label{fig:AveForestIV}
\centering
\includegraphics[width=0.5\linewidth]{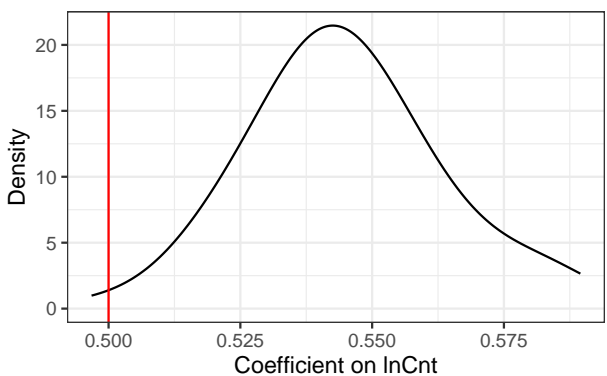}
\caption{Distribution of ForestIV Estimation on $lnCnt$ across 100 Simulation Runs}
\end{figure}

We can see that, in this case, while averaging across multiple estimates does indeed continue to mitigate the bias to some extent, it is nonetheless not as effective as our baseline ForestIV approach, which employs the single tuple that minimizes the empirical MSE. The limited utility of averaging is because the estimates are not ``symmetrically" distributed around the true values of coefficients (see Figure \ref{fig:AveForestIV} for the distribution of the averaged coefficient estimate on $lnCnt$). Rather, the distribution of values systematically deviates from the true values in a single direction. This is unsurprising, to some extent, given that 2SLS estimation is known to be biased in finite samples \citep{nagar1959bias,buse1992bias}.  Consequently, averaging produces worse correction results than picking the ``best'' tuple (the one that minimizes the empirical MSE). Note it is also possible that some tuples may ultimately fail to be rejected by the Hotelling's $T^2$ test, despite containing invalid or weak instruments (recall that a failure to reject the null hypothesis of equivalence does not imply that said null should be ``accepted''). Intuitively, by choosing to focus on the ``best'' tuple, we are applying the most stringent $p$-value threshold possible, and thus mitigating the potential that we unintentionally retain poor instruments in our final estimation. To the extent that retained instruments are of low quality, they may bias the resulting estimations towards that of the biased OLS \citep{murray2006avoiding,wooldridge2002}, which is exactly what we observe here. Nonetheless, we believe future work can investigate whether averaging can be advantageous under certain conditions, or develop new methodologies to derive more robust estimations when some instruments may be invalid or weak \citep[e.g., based on the work of ][]{conley2012plausibly}.

\clearpage

\section{Simulation Results for Binary Endogenous Covariate with Bank Marketing Data} \label{BankResults}

\subsection{Main Results} \label{Bank_MainResults}

In Table \ref{table:BankData}, we report the average coefficients and standard errors (in parentheses) across all simulation rounds for the biased regression and the unbiased regression. For ForestIV, we report the average coefficients and standard deviations of the sampling distributions across 100 simulation rounds as standard errors. We again compute the $p$-value (in square brackets) associated with a $t$-test \textit{comparing each estimated coefficient with its underlying true value}. We can see that directly using the random forest's predictions in the regression leads to 45.6\% underestimation on $Deposit$ on average. ForestIV is again effective in mitigating the estimation biases. ForestIV corrected coefficients are not significantly different from their true values, and ForestIV estimates achieve smaller standard errors (i.e., higher estimation precision) than the unbiased estimates.

\begin{table}[tbhp]
\label{table:BankData}
    \centering
    \begin{tabular}{c c c c c}
    \hline
         & True & Biased & Unbiased & ForestIV \\ 
         \hline
        Intercept       & 1.0 & 1.042 (0.010)   & 1.013 (0.055) &  0.995 (0.017) \\
                        &     & [0.007]     & [0.454]   & [0.751]   \\
        $Deposit$       & 0.5 & 0.272 (0.040)   & 0.519 (0.140) &  0.516 (0.116) \\
                        &     & [0.000]     & [0.426]   & [0.888]   \\
        $Z_1$           & 2.0 & 2.002 (0.017)   & 1.999 (0.090) &  2.001 (0.016) \\
                        &     & [0.504]     & [0.581]   & [0.920]   \\
        $Z_2$           & 1.0 & 1.001 (0.010)   & 1.004 (0.052) &  1.000 (0.010) \\
                        &     & [0.475]     & [0.454]   & [0.941]   \\
        \hline
        Ave\_MSE        &     & 0.108 & 0 & 0.023 \\
    \hline
    \end{tabular}
    \caption{ForestIV Results on Bank Marketing Data. Standard errors in parentheses. $p$-values comparing estimates with true values in square brackets. Ave\_MSE contains the average empirical MSE associated with each set of estimates across 100 simulation runs.}
\end{table}

\subsection{Sensitivity Analyses} \label{Bank_Sensitivity}
We conduct additional analyses to understand the performance of ForestIV for the case of a binary endogenous covariate with misclassification. We repeat all the sensitivity analyses that have been carried out for the continuous case and report the results in this subsection. 

The sensitivity analyses again consist of three parts. Respectively, we examine the performance of ForestIV estimations with respect to (1) the size of unlabeled data, (2) the amount of noise in econometric data, and (3) the predictive performance of the random forest (operationalized by changing the total number of trees). For all simulations in this section, we use the same "Bank Marketing" dataset \citep{moro2014data} and, unless otherwise noted, the simulation setups are the same as in the main manuscript.

\subsubsection{Size of Unlabeled Data} \label{Bank_Sensitivity_Size}

We repeat the simulation with 8 different sizes of unlabeled data, respectively 100, 500, 1,000, 10,000, 20,000, 30,000, 40,000, and 43,211 (i.e., all remaining instances). In Figure \ref{fig:BinaryUnlabelSize}, we plot the 95\% confidence interval and point estimate of $Deposit$ with respect to different sizes of unlabeled data. The 95\% confidence interval of ForestIV estimate is constructed as the range between the 2.5th percentile and 97.5th percentile of the sampling distribution.

\begin{figure}[tbhp]
\label{fig:BinaryUnlabelSize}
\centering
\includegraphics[width=0.8\linewidth]{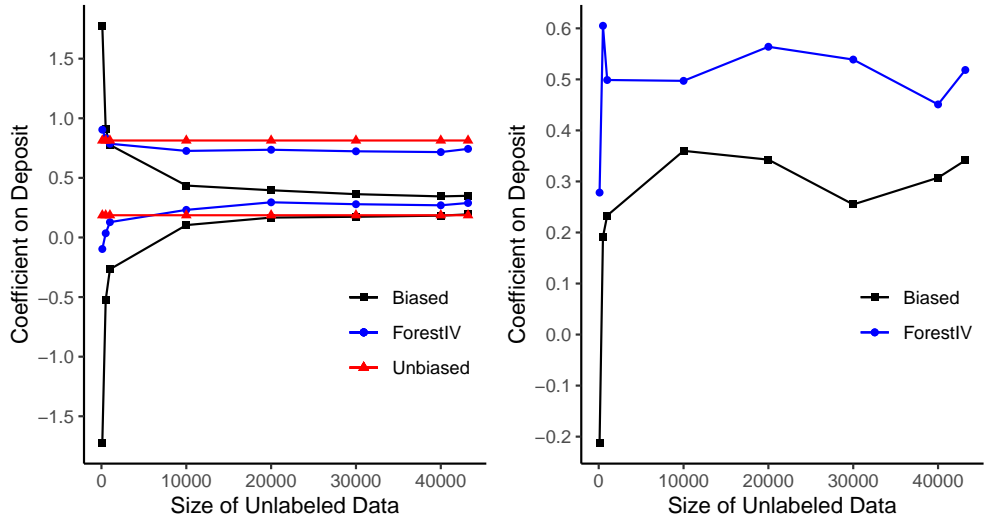}
\caption{Confidence Interval and Point Estimate on $Deposit$ with Different Sizes of Unlabeled Data. The left panel shows the 95\% confidence interval of the estimation on $Deposit$ with different sizes of unlabeled data. The right panel shows the point estimate on $Deposit$ with different sizes of unlabeled data based on a single simulation run. The true coefficient on $Deposit$ is 0.5.}
\end{figure}

We observe the same patterns as in the case of a continuous endogenous covariate. First, the biased estimate stays biased regardless of the size of unlabeled data. On the other hand, ForestIV produces converging estimation on $Deposit$ when sufficient unlabeled data is available. Having more unlabeled data can further shrink the confidence interval on $Deposit$, suggesting improvement in estimation precision. Meanwhile, the point estimates from a single simulation run, despite some local fluctuations, support the same observations.

\subsubsection{Noise in Econometric Data} \label{Bank_Sensitivity_Noise}
We introduce higher noise in econometric data by increasing $\sigma_{\varepsilon}$ from 2 to 5. The simulation results are summarized in the following Table \ref{table:BankNoise}.

\begin{table}[tbhp]
\label{table:BankNoise}
    \centering
    \begin{tabular}{c c c c c}
    \hline
         & True & Biased & Unbiased & ForestIV \\ 
         \hline
        Intercept       & 1.0 & 1.041 (0.025)    & 1.030 (0.138)  &  0.994 (0.038) \\
                        &     & [0.229]     & [0.483]   & [0.870]   \\
        $Deposit$       & 0.5 &  0.277 (0.099)   & 0.483 (0.346)  &  0.513 (0.206) \\
                        &     & [0.115]     & [0.433]   & [0.950]   \\
        $Z_1$           & 2.0 & 1.997 (0.042)    & 1.999 (0.224)  &  1.997 (0.040) \\
                        &     & [0.524]     & [0.558]   & [0.942]   \\
        $Z_2$           & 1.0 & 1.002 (0.024)    & 1.013 (0.129)  &  1.002 (0.024) \\
                        &     & [0.480]     & [0.491]   & [0.928]   \\
        \hline
        Ave\_MSE        &     & 0.317 & 0 & 0.193 \\
    \hline
    \end{tabular}
    \caption{ForestIV Results on Bank Marketing Data with $\sigma_{\varepsilon}=5$. Standard errors in parentheses. $p$-values comparing estimates with true values in square brackets. Ave\_MSE contains the average empirical MSE associated with each set of estimates across 100 simulation runs.}
\end{table}

We observe that, while the average point estimate on $Deposit$ is almost the same as when $\sigma_{\varepsilon}=2$, the associated average standard error becomes larger, as a direct result of the increased noise in data. Finally, ForestIV estimates again maintain clear increase in estimation precision over the unbiased estimates.

\subsubsection*{Predictive Performance of Random Forest} \label{Bank_Sensitivity_NTree}
We repeat the simulation with 5 different choices of the total number of trees in the random forest, $M$: 25, 50, 100, 150, and 200. We report the average ForestIV estimates across 100 simulation rounds for each choice of $M$ in the following Table \ref{table:BankNTree}. Predictive performance of the random forest is measured as the average accuracy. In addition, we report the average Hotelling $T^2$-statistic that compares $\widehat{\boldsymbol{\beta}}_{label}$ with ForestIV estimates in the last row of the table.

\begin{table}[tbhp]
\label{table:BankNTree}
    \centering
    \small
    \begin{tabular}{c c c c c c c}
    \hline
        &  & $M=25$ & $M=50$ & $M=100$ & $M=150$ & $M=200$ \\ 
        \hline
        Accuracy  &  & 89.67\% & 89.73\% & 89.86\% & 89.85\% & 89.89\%  \\
        \hline
        & True & & & & \\
        \hline
            Intercept & 1.0 & 0.993 (0.014) & 0.996 (0.015)  & 0.995 (0.017) & 0.995 (0.016) & 0.995 (0.014)    \\
                        &   & [0.656]  & [0.801]   & [0.751]  & [0.738]  & [0.753]     \\
            $Deposit$ & 0.5 & 0.521 (0.089) & 0.506 (0.099)  & 0.516 (0.116) & 0.509 (0.116) & 0.505 (0.123)    \\
                        &   & [0.810]  & [0.952]   & [0.888]  & [0.934]  & [0.966]     \\
            $Z_1$     & 2.0 & 2.001 (0.016) & 1.996 (0.020)  & 2.001 (0.017) & 2.002 (0.017) & 2.001 (0.017)    \\
                        &   & [0.945]  & [0.856]   & [0.920]  & [0.908]  & [0.926]     \\
            $Z_2$     & 1.0 & 1.001 (0.010) & 1.002 (0.009)  & 1.001 (0.010) & 0.999 (0.008) & 0.999 (0.010)    \\
                        &   & [0.905]  & [0.798]   & [0.941]  & [0.894]  & [0.933]     \\
            \hline
            Ave\_MSE  &  & 0.026 & 0.022 & 0.023 & 0.023 & 0.022 \\
        \hline
        Ave. Hotelling $T^2$ stats & & 3.0614 & 3.0786 & 2.9300 & 2.9404 & 2.6267 \\
        \hline
    \end{tabular}
    \caption{ForestIV Estimates for Different Choices of $M$. Standard errors in parentheses. $p$-values comparing estimates with true values in square brackets. Ave\_MSE contains the average empirical MSE associated with each set of estimates across 100 simulation runs.}
\end{table}

In this particular example, different choices of $M$ are associated with very similar accuracy on average. Nonetheless, we observe consistent patterns as in the continuous case. Specifically, taking the average point estimates as indications of ForestIV correction performance, $M=200$ produces relatively better corrections on this particular dataset. Having too few trees ($M=25$) hurts both the predictive accuracy and the average correction performance. Note that the correction performance again aligns with the average Hotelling $T^2$ statistic, with smaller statistic indicating better correction performance. 

\section{Benchmarking ForestIV with Three Alternative Approaches} \label{Benchmarking}

In the econometrics literature, researchers have proposed alternative bias correction methods to correct for estimation biases caused by measurement error, such as method-of-moments, likelihood-based approaches, deconvolution, regression calibration, and simulation-extrapolation \citep[we refer to][for detailed discussions of these methods]{grace2016statistical}. One common theme underlying these bias correction methods is that they all rely on information about the statistical properties of the measurement error, such as its distribution or moments. While this information is typically unobservable in most scenarios, it is available when the measurement error comes from the prediction error in a machine learning model \citep{yang2018mind}. This is because the process of building a predictive machine learning model typically involves evaluation of its performance on hold-out data, a process that yields performance metrics that quantify the degree of prediction error (and hence measurement error).

In this section, we focus on benchmarking ForestIV against three alternative bias correction approaches: (i) simulation-extrapolation (SIMEX), (ii) Latent Instrumental Variables (LatentIV), and (iii) regression adjustment for nonparametrically generated regressors. We provide a brief description of each method below.

SIMEX \citep{Cook1994} is a general simulation-based approach that can be applied to address measurement error in any econometric model. It directly leverages error magnitude information (e.g., error variance, based on the performance of machine learning predictions) to create a set of bootstrap samples of the observed data, artificially introducing larger measurement error with each subsequent re-sampling. The algorithm then estimates a corresponding set of coefficients with respect to different degrees of measurement errors, fits a parametric function to the pairs of coefficient-error observations, and finally extrapolates the coefficient estimates to the case where measurement error is zero. While the original SIMEX method was proposed to address measurement error in a continuous covariate, researchers later developed a variation named Misclassification-SIMEX (MC-SIMEX) to deal with misclassification in a discrete covariate \citep{Kuchenhoff2006,Kuchenhoff2007}. Flexibility is a key advantage of SIMEX. It only requires aggregated information about the measurement error, e.g., error variance for a continuous covariate or recall rates for a binary covariate, which can be estimated from the testing data. Moreover, SIMEX can be applied to a large set of econometric models with a standard procedure, and does not need to be explicitly re-formulated for each econometric model specification. The effectiveness of SIMEX for correcting biases due to measurement error has been comprehensively documented by \cite{yang2018mind}, for a variety of estimators and econometric specifications. Several studies have also discussed the application of SIMEX to more complicated measurement error problems, such as GLM with error-prone fixed and random effects \citep{Wang1998}, as well as nonparametric models \citep{Carroll1999}. We benchmark ForestIV respectively against SIMEX (a continuous covariate with measurement error) and against MC-SIMEX (a binary covariate with misclassification) on two different datasets. 

The LatentIV approach is proposed by \cite{ebbes2005solving,ebbes2009frugal} to address endogeneity in linear regression models without the use of observable instruments. LatentIV achieves identification by modeling a latent (i.e., unobserved) discrete instrument to account for the correlation between the endogenous covariate and the regression model's error term. We benchmark ForestIV against LatentIV using the Bike Sharing dataset.

Finally, we consider a particular regression adjustment approach proposed in the generated regressors literature. \cite{meng2016linear} studies the estimation of a linear regression where one covariate is not directly observed, but can be approximated based on a sample of relevant data (e.g., a nation's income inequality measure is unobserved but can be estimated from a sample of individual incomes in the nation). Because sample-based approximation of the unobserved covariate can be noisy, the linear regression suffers from a measurement error issue. The authors assume that the functional relationship underlying the generated regressor is inconsistent across observations, and thus the generated regressor can be viewed as nonparametric. \cite{meng2016linear} derive an explicit formula for the magnitude of bias, as a function of the first two \textit{moments} of the measurement error (e.g., mean and variance), which allows the biased estimates to be adjusted accordingly. In our setting, the moment statistics of measurement/prediction error can be readily estimated using the testing data, where prediction errors are directly observed. We are therefore able to implement \cite{meng2016linear}'s proposed adjustment approach in our case as well, and benchmark the performance against ForestIV using the same Bike Sharing dataset.

\subsection{Benchmarking with SIMEX: Continuous Case} \label{Benchmarking_SIMEX_Continuous}

The benchmarking simulation experiments are set up as follows. We use the Boston Housing Dataset \citep{harrison1978hedonic} to build a numeric prediction model. The Boston Housing Dataset contains 14 attributes describing housing located in 506 census tracts in the Boston area. This data has frequently been used as a benchmarking dataset for numeric prediction algorithms \citep[e.g.,][]{rose1998deterministic,lim2000comparison}. We build a random forest, containing 100 trees, to predict the median value of houses in a census tract (\textit{MEDV}), based on 13 different features, e.g., average number of rooms, property tax. We take a random sample of 200 tracts as $D_{train}$ and hold out an additional random sample of 50 tracts to serve as $D_{test}$. The random forest model is then used to predict the median home values for the remaining 256 tracts. We denote the aggregate predictions from the random forest as $\widehat{MEDV}$, and predictions from each individual tree $i \in \{1,\ldots,100\}$ as $\widehat{MEDV}_i$.

For the econometric model to be estimated, we simulate an artificial dataset ($N=506$) where \textit{MEDV} is an independent covariate. We incorporate two other control variables, $Z_1$ and $Z_2$, into the model. Specifically, $Z_1 \sim Bernoulli(0.6)$ is a dummy variable that takes on a value of 1 with 60\% probability, and $Z_2 \sim N(0,1)$ is a normally distributed continuous variable. We also simulate an error term, $\varepsilon \sim N(0,0.1^2)$. The dependent variable is generated as a linear combination of all variables, $Y=1+0.5MEDV+2Z_1+Z_2+\varepsilon$. 

As before, the simulation is repeated for 100 rounds, and we report the average coefficients and standard errors across all rounds. For the sake of brevity, we refer to \cite{yang2018mind} for technical details on applying the SIMEX approach to obtain corrected coefficients that account for measurement error in $\widehat{MEDV}$. The results are summarized in Table \ref{table:SIMEX}. In addition, in Figure \ref{fig:SIMEXDist} we plot the distributions of the coefficient associated with $MEDV$ across all simulation runs.

\begin{table}[tbhp]
\label{table:SIMEX}
    \centering
    \begin{tabular}{c c c c c}
    \hline
         & True & Biased & SIMEX & ForestIV \\ 
         \hline
        Intercept       & 1.0 & -0.745 (0.402)   & 1.154 (0.382)  &  0.999 (0.245) \\
                        &     & [0.027]     & [0.689]   & [0.997]   \\
        $MEDV$          & 0.5 &  0.579 (0.016)   & 0.494 (0.060)  &  0.500 (0.010) \\
                        &     & [0.024]     & [0.707]   & [0.999]   \\
        $Z_1$           & 2.0 & 1.961 (0.236)    & 1.958 (0.260)  &  1.986 (0.171) \\
                        &     & [0.458]     & [0.871]   & [0.935]   \\
        $Z_2$           & 1.0 & 0.989 (0.115)    & 0.991 (0.126)  &  1.001 (0.086) \\
                        &     & [0.461]     & [0.943]   & [0.991]   \\
        \hline
        Ave\_MSE        &     & 3.284 & 0.259 & 0.097 \\
    \hline
    \end{tabular}
    \caption{Estimation Results using SIMEX and ForestIV. Standard errors in parentheses. $p$-values comparing estimates with true values in square brackets. Ave\_MSE contains the average empirical MSE associated with each set of estimates across 100 simulation runs.}
\end{table}

As shown in Table \ref{table:SIMEX}, directly using the $\widehat{MEDV}$ predicted by random forest in the regression model results in significant biases. The coefficient on $MEDV$ is \textit{overestimated} by about 15.8\% on average. Moreover, across simulation runs, the coefficient on $MEDV$ may be biased both upward and downward (see Figure \ref{fig:SIMEXDist}). Coefficients associated with the other covariates, $Z_1$ and $Z_2$, are also biased to different degrees and in different directions, although the average biases are not large. Second, our ForestIV approach effectively mitigates estimation biases. It almost fully recovered the unbiased coefficient on all regressors. Third, compared with SIMEX, ForestIV produced point estimates that are closer to the true values. In addition, as shown in Figure \ref{fig:SIMEXDist}, the distribution of the coefficient on $MEDV$ recovered by ForestIV exhibits smaller standard deviation than that recovered by SIMEX, indicating greater stability as well.

\begin{figure}[tbhp]
\label{fig:SIMEXDist}
\centering
\includegraphics[width=0.6\linewidth]{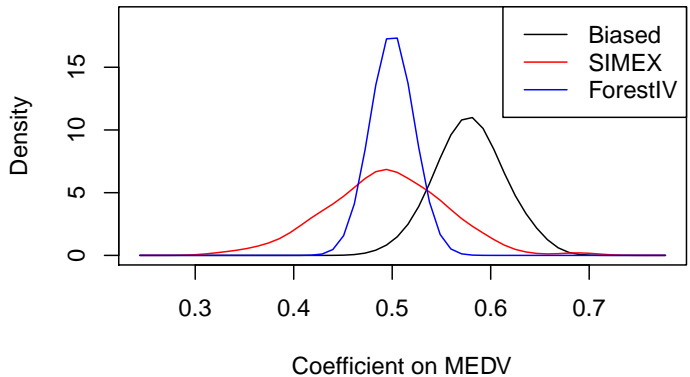}
\caption{Distributions of Biased and Corrected Coefficients on $MEDV$ across Simulation Runs }
\end{figure}

\subsection{Benchmarking with SIMEX: Binary Case} \label{Benchmarking_SIMEX_Binary}

We use the Wisconsin Breast Cancer Dataset \citep{wolberg1990multisurface}. The dataset contains 683 clinical cases of breast cancer diagnoses. We build a random forest of 100 trees to predict the cancer outcome (benign or malignant) based on 9 biological features. Similar to the previous simulation setup, we take a random sample of 200 cases as $D_{train}$, 50 cases as $D_{test}$, and the rest 433 cases as $D_{unlabel}$. Next, we simulate an econometric model: $Y=1+0.5Cancer+2Z_1+Z_2+\varepsilon$ ($N=683$), where $Z_1 \sim Uniform[-1,1]$, $Z_2 \sim N(0,1^2)$, and $\varepsilon \sim N(0,0.1^2)$. 

As before, the simulation is repeated for 100 rounds, and we report the average coefficients and standard errors across all rounds. We again refer to \cite{yang2018mind} for technical details on applying the MC-SIMEX approach. The results are summarized in Table \ref{table:SIMEX_binary}. In Figure \ref{fig:SIMEXDist_binary}, we plot the distributions of the coefficient associated with $Cancer$ across all simulation runs.

\begin{table}[tbhp]
\label{table:SIMEX_binary}
    \centering
    \begin{tabular}{c c c c c}
    \hline
         & True & Biased & MC-SIMEX & ForestIV \\ 
         \hline
        Intercept       & 1.0 & 1.013 (0.008)  & 1.005 (0.008)  &  1.004 (0.008) \\
                        &     & [0.234]   & [0.405]   & [0.505]   \\
        $Cancer$        & 0.5 & 0.463 (0.014)  & 0.519 (0.059)  &  0.496 (0.012) \\
                        &     & [0.044]   & [0.264]   & [0.739]   \\
        $Z_1$           & 2.0 & 2.001 (0.011)  & 2.001 (0.012)  &  2.000 (0.009) \\
                        &     & [0.492]   & [0.934]   & [0.999]   \\
        $Z_2$           & 1.0 & 0.999 (0.006)  & 0.999 (0.007)  &  0.999 (0.006) \\
                        &     & [0.506]   & [0.868]   & [0.886]   \\
        \hline
        Ave\_MSE        &     & 0.0028 & 0.0051 & 0.0009 \\
    \hline
    \end{tabular}
    \caption{Estimation Results using MC-SIMEX and ForestIV. Standard errors in parentheses. $p$-values comparing estimates with true values in square brackets. Ave\_MSE contains the average empirical MSE associated with each set of estimates across 100 simulation runs.}
\end{table}

Based on Table \ref{table:SIMEX_binary}, directly using random forest predictions $\widehat{Cancer}$ in the regression causes about 7.4\% underestimation on its coefficient. Both ForestIV and MC-SIMEX are able to mitigate the estimation biases, and ForestIV outperforms MC-SIMEX by producing average point estimates that are even closer to the underlying true values. In addition, Figure \ref{fig:SIMEXDist_binary} shows that ForestIV estimation on $Cancer$ again has a substantially smaller standard deviation than that from MC-SIMEX, indicating greater stability.

\begin{figure}[tbhp]
\label{fig:SIMEXDist_binary}
\centering
\includegraphics[width=0.6\linewidth]{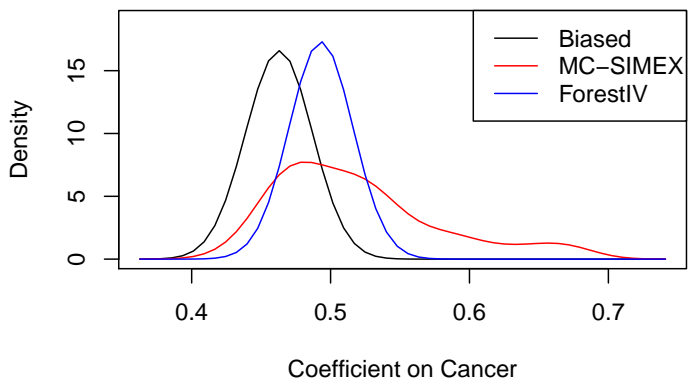}
\caption{Distributions of Biased and Corrected Coefficients on $Cancer$ across Simulation Runs}
\end{figure}

\subsection{Limitation of SIMEX and Robustness of ForestIV} \label{Benchmarking_SIMEX_Fail}

As part of this benchmarking exercise, we have also identified a previously undocumented yet common scenario in which SIMEX yields systematically poor results. In particular, when measurement error is correlated with another precisely-measured covariate in the econometric model, the SIMEX algorithm actually \textit{introduces} bias into the coefficient associated with the precisely-measured covariate. Importantly, the ForestIV approach does not suffer from this problem, and is thus likely to be preferable to SIMEX in such cases. 

While SIMEX is a generally reliable approach to dealing with measurement error and has been demonstrated to perform well in various types of econometric models \citep{yang2018mind}, we have identified an important limitation that causes SIMEX to produce problematic correction results under certain conditions.

Again, consider the regression model $Y \sim \widehat{X}\beta_X + \boldsymbol{Z\beta_Z}$, where $\widehat{X}$ is measured with additive error, i.e., $\widehat{X} = X+e$. Suppose that the measurement error component, $e$, is correlated with one of the precisely-measured control variables in the model, e.g., $\exists Z^* \in \boldsymbol{Z}, Cov(e,Z^*) \neq 0$. The SIMEX-corrected coefficient on $Z^*$ under this setup can be even more biased than it is in the absence of correction. 

This represents a realistic scenario that can arise when combining machine learning and econometric modeling. As an example, consider a credit scoring prediction model. An older person's credit score might be easier to predict (with less error) than a younger person's, as more data on historical consumption and repayment will be available in the former case. The error in credit score prediction may thus be correlated with age, which would plausibly appear as a control variable in many econometric models. Another example that has attracted heated discussion is Propublica's critique of the COMPAS risk tool \citep{angwin2016machine}, a predictive model used in the U.S. criminal justice system to assess a defendant's risk of recidivism. The prediction errors that COMPAS produces are correlated with race; the algorithm has been shown to be more inaccurate for (i.e., biased against) African-American defendants than for white defendants. Similar racial disparity in predictive performance has been documented in gender classification \citep{buolamwini2018gender}, where researchers found that image classifiers are less accurate for darker-skinned people.

In the following theorem, we prove this limitation of SIMEX in the simple case of a linear regression with additive independent (i.e., classical) measurement error. Formally, consider a population regression equation: $Y=\beta_0+\beta_1X_1+\beta_2X_2+\varepsilon$, where $X_1$ is measured with additive independent error, i.e., $\widehat{X_1}=X_1+e$, and $Cov(X_1,e)=0$. Suppose that $X_2$ is correlated with the measurement error, i.e., $Cov(X_2,e)=\sigma_{2e} \neq 0$. We still assume the model error term, $\varepsilon$, is exogenous, i.e., $Cov(X_1,\varepsilon)=Cov(X_2,\varepsilon)=Cov(e,\varepsilon)=0$. For simplicity, we also assume $Cov(X_1,X_2)=0$.\footnote{Relaxing this assumption makes the derivation more elaborate without changing the underlying mechanism. When $Cov(X_1,X_2) \neq 0$, our statement is still true under slightly more strict conditions} Further denote $Var(X_1)=\sigma_1^2$, $Var(X_2)=\sigma_2^2$, and $Var(e) = \sigma_e^2$. Finally, denote the true coefficient, the biased (i.e., uncorrected) coefficient, and the SIMEX-corrected coefficient on $X_2$ as $\beta_2$, $\widehat{\beta_2}$, and $\widehat{\beta_2}^{SIMEX}$ respectively.

\textbf{THEOREM}. \textit{$|\widehat{\beta_2}^{SIMEX} - \beta_2| > |\widehat{\beta_2} - \beta_2|$, i.e., the SIMEX-corrected coefficient will become even more biased than in the absence of correction, if and only if $\left|\sigma_2^2 \sigma_1^2 - \sigma_{2e}^2\right| < \left|\sigma_2^2 (\sigma_1^2 + \sigma_e^2) - \sigma_{2e}^2\right|$.}

\begin{proof}
Consider OLS estimation of the biased regression of $Y$ on $\{\widehat{X_1},X_2\}$. Using the regression anatomy method \citep{angrist2008mostly}, the estimated coefficient associated with $X_2$ is $\widehat{\beta_2}=\frac{Cov(Y,\widetilde{X_2})}{Var(\widetilde{X_2})}$, where $\widetilde{X_2}$ is the residual from regressing $X_2$ on $\widehat{X_1}$, i.e., $X_2=r_0+r_1\widehat{X_1}+\widetilde{X_2}$. We know that $r_1=\frac{Cov(X_2,\widehat{X_1})}{Var(\widehat{X_1})}$, and therefore $\widetilde{X_2}=X_2-r_0-\frac{Cov(X_2,\widehat{X_1})}{Var(\widehat{X_1})} \widehat{X_1}$. It follows that $\widehat{\beta_2}=\frac{Cov(Y,\widetilde{X_2})}{Var(\widetilde{X_2})} = \frac{Cov\left(\beta_0+\beta_1X_1+\beta_2X_2+\varepsilon, \quad X_2-r_0-\frac{Cov(X_2,\widehat{X_1})}{Var(\widehat{X_1})} \widehat{X_1}\right)}{Var\left(X_2-r_0-\frac{Cov(X_2,\widehat{X_1})}{Var(\widehat{X_1})} \widehat{X_1}\right)}$. Dropping the constant terms $\beta_0,r_0$ and the exogenous $\varepsilon$, the above expression simplifies to $\widehat{\beta_2}= \frac{Cov\left(\beta_1X_1+\beta_2X_2, \quad X_2-\frac{Cov(X_2,\widehat{X_1})}{Var(\widehat{X_1})} \widehat{X_1}\right)}{Var\left(X_2-\frac{Cov(X_2,\widehat{X_1})}{Var(\widehat{X_1})} \widehat{X_1}\right)}=\frac{-\beta_1 \frac{Cov(X_2,\widehat{X_1})}{Var(\widehat{X_1})} Cov(X_1, \widehat{X_1}) + \beta_2\sigma_2^2 - \beta_2 \frac{Cov^2(X_2,\widehat{X_1})}{Var(\widehat{X_1})}}{\sigma_2^2 - \frac{Cov^2(X_2,\widehat{X_1})}{Var(\widehat{X_1})}}=\beta_2 - \beta_1 \frac{Cov(X_2, \widehat{X_1}) Cov(X_1,\widehat{X_1})}{\sigma_2^2 Var(\widehat{X_1}) - Cov^2(X_2, \widehat{X_1})}$. Given that $Var(\widehat{X_1})=\sigma_1^2 + \sigma_e^2$, $Cov(X_1,\widehat{X_1})=\sigma_1^2$, $Cov(X_2, \widehat{X_1})=\sigma_{2e}$, we have $\widehat{\beta_2} = \beta_2 - \beta_1 \frac{\sigma_{2e}\sigma_1^2}{\sigma_2^2 (\sigma_1^2 + \sigma_e^2) - \sigma_{2e}^2}$. This implies that the absolute bias in the coefficient associated with $X_2$, $|\widehat{\beta_2} - \beta_2|=\left|\beta_1 \frac{\sigma_{2e}\sigma_1^2}{\sigma_2^2 (\sigma_1^2 + \sigma_e^2) - \sigma_{2e}^2} \right|$.

Now consider the SIMEX correction procedure \citep{Cook1994}. In the simulation step, SIMEX creates $\widehat{X_1}^{(\lambda)} = \widehat{X_1} + \sqrt{\lambda}z = X_1+e+\sqrt{\lambda}z$, where $z \sim N(0, \sigma_e^2)$, and thereby introduces more measurement error. Note that $Var(\widehat{X_1}^{(\lambda)}) = \sigma_1^2+(1+\lambda)\sigma_e^2$, and $Cov(X_2,\widehat{X_1}^{(\lambda)})=\sigma_{2e}$ because $z$ is independently generated. Following the same derivation above, we know that regressing $Y$ on $\{\widehat{X_1}^{(\lambda)},X_2\}$, we would have $\widehat{\beta_2}^{(\lambda)} = \beta_2 - \beta_1 \frac{\sigma_{2e}\sigma_1^2}{\sigma_2^2 \left(\sigma_1^2 + (1+\lambda)\sigma_e^2\right) - \sigma_{2e}^2}$, or equivalently, $|\widehat{\beta_2}^{(\lambda)} - \beta_2| = \left|\beta_1 \frac{\sigma_{2e}\sigma_1^2}{\sigma_2^2 \left(\sigma_1^2 + (1+\lambda)\sigma_e^2\right) - \sigma_{2e}^2}\right|$. In the extrapolation step, SIMEX estimates $\widehat{\beta_2}^{(-1)}$, i.e., the coefficient that would be obtained had there been no measurement error (note that $\widehat{\beta_2}^{(-1)} \equiv \widehat{\beta_2}^{SIMEX}$). Accordingly, $|\widehat{\beta_2}^{(-1)} - \beta_2| = \left|\beta_1 \frac{\sigma_{2e}\sigma_1^2}{\sigma_2^2 \sigma_1^2 - \sigma_{2e}^2}\right|$. 

Finally, we compare $|\widehat{\beta_2}^{(-1)} - \beta_2|=\left|\beta_1 \frac{\sigma_{2e}\sigma_1^2}{\sigma_2^2 \sigma_1^2 - \sigma_{2e}^2}\right|$ and $|\widehat{\beta_2} - \beta_2|=\left|\beta_1 \frac{\sigma_{2e}\sigma_1^2}{\sigma_2^2 (\sigma_1^2 + \sigma_e^2) - \sigma_{2e}^2} \right|$, or equivalently, compare $\frac{1}{\left|\sigma_2^2 \sigma_1^2 - \sigma_{2e}^2\right|}$ and $\frac{1}{\left|\sigma_2^2 (\sigma_1^2 + \sigma_e^2) - \sigma_{2e}^2\right|}$. Therefore, the condition $\left|\sigma_2^2 \sigma_1^2 - \sigma_{2e}^2\right| < \left|\sigma_2^2 (\sigma_1^2 + \sigma_e^2) - \sigma_{2e}^2\right|$ implies that $\frac{1}{\left|\sigma_2^2 \sigma_1^2 - \sigma_{2e}^2\right|} > \frac{1}{\left|\sigma_2^2 (\sigma_1^2 + \sigma_e^2) - \sigma_{2e}^2\right|} \Rightarrow |\widehat{\beta_2}^{(-1)} - \beta_2| > |\widehat{\beta_2} - \beta_2|$, i.e., the SIMEX corrected coefficient on $X_2$ becomes even more biased than it would be in the absence of correction.
\end{proof}

\textbf{Remark.} We note that the condition $\left|\sigma_2^2 \sigma_1^2 - \sigma_{2e}^2\right| < \left|\sigma_2^2 (\sigma_1^2 + \sigma_e^2) - \sigma_{2e}^2\right|$ is quite easily satisfied. The right-hand-side of the inequality is equivalent to $\left|\sigma_2^2 \sigma_1^2 - \sigma_{2e}^2 + \sigma_2^2\sigma_e^2\right|$. Therefore, a sufficient (but not necessary) condition for the inequality to hold is $\sigma_2^2 \sigma_1^2 - \sigma_{2e}^2 \geq 0$. Because $\sigma_2^2 \sigma_1^2 - \sigma_{2e}^2 \geq 0 \Leftrightarrow \frac{\sigma_1^2}{\sigma_e^2} \geq \frac{\sigma_{2e}^2}{\sigma_2^2 \sigma_e^2} = \rho_{2e}^2$, we can see that the inequality always holds if $\sigma_1^2 \geq \sigma_e^2$, which means that the variance of measurement error is no larger than the variance of true covariate. This is typically true unless the measurement error is exceedingly large.

Essentially, the SIMEX approach relies on the implicit assumption that degree of measurement error is \textit{positively} related to degree of bias. This assumption is violated when a precisely-measured covariate is correlated with the measurement error, as we have shown in the above theorem. As a result, unless special modifications are made to the SIMEX procedure, it produces incorrect results on the precisely-measured covariate. However, our ForestIV approach does not suffer from this issue, because it does not rely on the same implicit assumption. Instead, the identified instruments should mitigate estimation bias on the error-prone covariate, without introducing additional bias to the estimates of precisely-measured covariates.

We empirically demonstrate this limitation of SIMEX and the robustness of ForestIV via another set of simulations, using the Boston Housing dataset. The basic setup is the same as before, with one change: one control variable, $Z_2 \sim N(0,1)$, is generated so that it is correlated with the prediction error in $\widehat{MEDV}$ (the aggregate predictions from the random forest model), with a correlation coefficient of 0.3. We repeat the same set of regression analyses and report the results in Table \ref{table:SIMEXFail}. 

\begin{table}[tbhp]
\label{table:SIMEXFail}
    \centering
    \begin{tabular}{c c c c c}
    \hline
         & True & Biased & SIMEX & ForestIV \\ 
         \hline
        Intercept       & 1.0 & -0.962 (0.399)   & 0.688 (0.385)  &  0.840 (0.384) \\
                        &     & [0.002]     & [0.418]   & [0.677]   \\
        $MEDV$          & 0.5 & 0.588 (0.016)    & 0.513 (0.060)  &  0.506 (0.015) \\
                        &     & [0.004]     & [0.416]   & [0.689]   \\
        $Z_1$           & 2.0 & 2.017 (0.232)    & 2.011 (0.258)  &  2.017 (0.219) \\
                        &     & [0.458]     & [0.966]   & [0.938]   \\
        $Z_2$           & 1.0 & 0.483 (0.114)    & 0.401 (0.155)  &  0.837 (0.126) \\
                        &     & [0.000]     & [0.000]   & [0.196]   \\
        \hline
        Ave\_MSE        &     & 4.351 & 0.699 & 0.264 \\
    \hline
    \end{tabular}
    \caption{Estimation Results using SIMEX and ForestIV with $\rho_{Z_2, e} = 0.3$. Standard errors in parentheses. $p$-values comparing estimates with true values in square brackets. Ave\_MSE contains the average empirical MSE associated with each set of estimates across 100 simulation runs.}
\end{table}

Based on Table \ref{table:SIMEXFail}, we can see that, when the measurement error is correlated with the precisely-measured control variable $Z_2$, the coefficient on $Z_2$ is drastically underestimated by 51.7\%, and, consistent with our theoretical result, the SIMEX-corrected coefficient becomes even more biased. However, our proposed ForestIV approach is able to not only mitigate the bias on $MEDV$, but also notably correct the coefficient on $Z_2$ in the right direction. This suggests that ForestIV is more robust than SIMEX when the measurement error is correlated with certain control variables in the model.

\subsection{Benchmarking with LatentIV} \label{Benchmarking_LatentIV}
Next, we benchmark our proposed ForestIV with the LatentIV approach on the Bike Sharing dataset. We adopt the implementation of LatentIV in the R package ``REndo". Note that this implementation of LatentIV does not support the estimation of other exogenous control variables in the model. Therefore, we drop the control variables from the simulation, and the dependent variable is simulated simply as $Y=1+0.5lnCnt+\varepsilon$, where $\varepsilon \sim N(0,4)$. We apply both ForestIV and LatentIV on this data. The average estimates from LatentIV across 100 simulation runs, together with the biased, unbiased, and ForestIV estimates, are reported in Table \ref{table:BikeData_LatentIV}.

\begin{table}[tbhp]
\label{table:BikeData_LatentIV}
    \centering
    \begin{tabular}{c c c c c c}
    \hline
         & True & Biased & Unbiased & ForestIV & LatentIV \\ 
         \hline
        Intercept       & 1.0 & 0.711 (0.062)   & 0.972 (0.203) &  0.938 (0.148) & 0.776 (0.733)\\
                        &     & [0.001]     & [0.568]   & [0.677]   & [0.760]\\
        $lnCnt$         & 0.5 &  0.564 (0.013)  & 0.507 (0.040) &  0.515 (0.029) & 0.549 (0.166)\\
                        &     & [0.002]     & [0.538]   & [0.613]  & [0.770]\\
        \hline
        Ave\_MSE        &     & 0.112 & 0 & 0.010 & 1.346\\
    \hline
    \end{tabular}
    \caption{Benchmarking ForestIV with LatentIV on Bike Sharing Data. Standard errors in parentheses. $p$-values comparing estimates with true values in square brackets. Ave\_MSE contains the average empirical MSE associated with each set of estimates across 100 simulation runs.}
\end{table}

We find that LatentIV is able to mitigate the estimation biases on $lnCnt$ and the intercept term to some degree. Nonetheless, the method appears to be systematically less effective than ForestIV. In addition, the standard errors of the LatentIV estimates are much larger than those yielded by ForestIV, indicating potential instability of the results.

\subsection{Benchmarking with Generated Regressor Adjustment} \label{Benchmarking_RegAdj}
Finally, we benchmark ForestIV against the regression adjustment approach developed by \cite{meng2016linear} for measurement error in nonparametrically generated regressors. We apply the formula in \citep[][p.305]{meng2016linear} to obtain the adjusted coefficient and standard error on $lnCnt$.\footnote{\cite{meng2016linear} does not explicitly discuss how to adjust for bias in the intercept term. We therefore report the same intercept term as the biased regression.} The average estimates from the regression adjustment across 100 simulation runs, together with the biased, unbiased, and ForestIV estimates, are reported in Table \ref{table:BikeData_RegAdj}.

\begin{table}[tbhp]
\label{table:BikeData_RegAdj}
    \centering
    \begin{tabular}{c c c c c c}
    \hline
         & True & Biased & Unbiased & ForestIV & Generated Regressor Adjustment \\ 
         \hline
        Intercept       & 1.0 & 0.702 (0.063)   & 1.018 (0.204) &  0.957 (0.134) & 0.702 (0.063)\\
                        &     & [0.004]     & [0.511]   & [0.745]   & [0.004]\\
        $lnCnt$         & 0.5 &  0.566 (0.013)  & 0.498 (0.040) &  0.512 (0.027) & 0.563 (0.147)\\
                        &     & [0.002]     & [0.530]   & [0.652]  & [0.664]\\
        $Z_1$           & 2.0 & 2.000 (0.003)   & 1.999 (0.011) &  2.000 (0.003) &  2.000 (0.003)\\
                        &     & [0.459]     & [0.524]   & [0.977]  & [0.459] \\
        $Z_2$           & 1.0 & 1.000 (0.002)   & 0.999 (0.006) &  1.000 (0.002) & 1.000 (0.002)\\
                        &     & [0.480]     & [0.486]   & [0.989]  & [0.480]\\
        \hline
        Ave\_MSE        &     & 0.150 & 0 & 0.017 & 0.164\\
    \hline
    \end{tabular}
    \caption{Benchmarking ForestIV with Generated Regressor Adjustment on Bike Sharing Data. Standard errors in parentheses. $p$-values comparing estimates with true values in square brackets. Ave\_MSE contains the average empirical MSE associated with each set of estimates across 100 simulation runs.}
\end{table}

We find that the generated regressor adjustment approach only slightly reduces the bias on $lnCnt$, and is clearly less effective than ForestIV. We believe that the limitations of the regression adjustment approach in this case are potentially attributable to the fact that it \textit{only} exploits distributional information on the mean and variance of measurement error, essentially ignoring other distributional information. As a final note, we acknowledge that alternative methodologies which have seen recent theoretical development in the generated regressors literature (or the econometrics literature at large) may achieve better bias correction results than ForestIV in certain situations. As such, we believe this to be a potentially quite fruitful future research direction by which our method could be improved upon.

\clearpage

\section{Proofs of Theoretical Results} \label{Proofs}

\mainTheorem*

\begin{proof}
Adopting notations from \cite{Breiman2001}, as the number of trees in the random forest goes to infinity, the generalization error of the forest is denoted as $PE(forest) = \lim_{M \rightarrow \infty} \mathbb{E}_{\mathbf{f}}[\widehat{X} - X]^2$. Next, we restate two known results as lemmas.
\begin{lemma}{\citep[][Theorem 11.2.]{Breiman2001}}
 $PE(forest) = \mathbb{E}_i\mathbb{E}_j\mathbb{E}_{\mathbf{f}} Cov(e^{(i)}, e^{(j)})$.
\end{lemma}

\begin{lemma}{\citep[][Theorem 1]{scornet2015consistency}}
Under Assumptions 1-2 , $\lim_{n \rightarrow \infty} PE(forest) = 0$.
\end{lemma}

The two lemmas collectively imply that $\lim_{n \rightarrow \infty} \mathbb{E}_i\mathbb{E}_j\mathbb{E}_{\mathbf{f}} Cov(e^{(i)}, e^{(j)}) = 0$. It follows that 
\begin{align*}
    & \lim_{n \rightarrow \infty} \mathbb{E}_i\mathbb{E}_j\mathbb{E}_{\mathbf{f}} Cov(e^{(i)}, e^{(j)}) = 0 \\
    \Leftrightarrow & \lim_{n \rightarrow \infty} \mathbb{E}_i\mathbb{E}_j\mathbb{E}_{\mathbf{f}} Cov\left(e^{(i)}, (\widehat{X}^{(j)} - X)\right) = 0 \\
    \Leftrightarrow & \lim_{n \rightarrow \infty} \mathbb{E}_i\mathbb{E}_j\mathbb{E}_{\mathbf{f}} Cov(e^{(i)}, \widehat{X}^{(j)}) - \lim_{n \rightarrow \infty} \mathbb{E}_i\mathbb{E}_{\mathbf{f}} Cov(e^{(i)}, X) = 0
\end{align*}

Based on Assumption 3 (classical measurement error), $\lim_{n \rightarrow \infty} \mathbb{E}_i\mathbb{E}_{\mathbf{f}} Cov(e^{(i)}, X) = 0$. Therefore, we have $\lim_{n \rightarrow \infty} \mathbb{E}_i\mathbb{E}_j\mathbb{E}_{\mathbf{f}} Cov(\widehat{X}^{(j)}, e^{(i)}) = 0$.
\end{proof}

\theoremBinaryA*

\begin{proof}
\cite{Breiman2001} proves that the error rate of a random forest decreases with $\mathbb{E}_j \mathbb{E}_i \left[ Corr(rmg(\widehat{X}^{(i)}), rmg(\widehat{X}^{(j)})) \right]$, where $rmg(\widehat{X}^{(i)})$ represents the \textit{raw marginal function} of tree $i$'s predictions. Under binary classification, the raw marginal function is defined as $rmg(\widehat{X}^{(i)}) = \mathbb{I}(\widehat{X}^{(i)} = X) - \mathbb{I}(\widehat{X}^{(i)} \neq X)$, where $\mathbb{I}$ is an indicator function that checks vectors $\widehat{X}^{(i)}$ and $X$ element-wise, and takes value 1 if the enclosed relationship is true or 0 otherwise. In other words, $\mathbb{I}(\widehat{X}^{(i)} = X)$ is a vector where correct predictions are marked with 1, and $\mathbb{I}(\widehat{X}^{(i)} \neq X)$ is a vector where incorrect predictions are marked with 1. Denote $\boldsymbol{1} = (1, \dots, 1)$ as a vector of 1 with the same length as the vector of predictions. Clearly, we have $\mathbb{I}(\widehat{X}^{(i)} = X) = \boldsymbol{1} - \mathbb{I}(\widehat{X}^{(i)} \neq X)$, and $\mathbb{I}(\widehat{X}^{(i)} \neq X) = |e^{(i)}|$. Therefore, we know $Corr((rmg(\widehat{X}^{(i)}), rmg(\widehat{X}^{(j)})) = Corr(\mathbb{I}(\widehat{X}^{(i)} = X) - \mathbb{I}(\widehat{X}^{(i)} \neq X), \mathbb{I}(\widehat{X}^{(j)} = X) - \mathbb{I}(\widehat{X}^{(j)} \neq X)) = Corr(\boldsymbol{1}-2\mathbb{I}(\widehat{X}^{(i)} \neq X), \boldsymbol{1}-2\mathbb{I}(\widehat{X}^{(j)} \neq X)) = Corr(\mathbb{I}(\widehat{X}^{(i)} \neq X), \mathbb{I}(\widehat{X}^{(j)} \neq X)) = Corr(|e^{(i)}|, |e^{(j)}|)$.
\end{proof}

\theoremBinaryB*

\begin{proof}
We prove this theorem for a given sample with size $N$. For notational simplicity, write the ground truth as $X = \{a_k\}_{k=1}^N$, and similarly write the prediction vector and error vector of tree $i$ as $\widehat{X}^{(i)} = \{p_{ik}\}_{k=1}^N$, $e^{(i)} = \{e_{ik}\}_{k=1}^N$. Suppose the number of data points where $a_k = \alpha$ and $p_{ik} = \beta$ is $n_{\alpha \beta}$ ($\alpha, \beta \in \{0,1\}$). Clearly, $n_{00} + n_{01} + n_{10} + n_{11} = N$, and the relationship between $X$ and $e^{(i)}$ is fully described as follows:

\begin{itemize}
    \item There are $n_{00}$ data points where $a_k=0$ and $e_{ik}=0$;
    \item There are $n_{01}$ data points where $a_k=0$ and $e_{ik}=1$;
    \item There are $n_{10}$ data points where $a_k=1$ and $e_{ik}=-1$;
    \item There are $n_{11}$ data points where $a_k=1$ and $e_{ik}=0$.
\end{itemize}

Next, write $Cov(e^{(i)}, X) = \dfrac{1}{N^2} (N \sum e_{ik} a_k - \sum e_{ik} \sum a_k)$. Note that $\sum e_{ik} a_k = -n_{10}$, $\sum e_{ik}=n_{01}-n_{10}$, and $\sum a_k=n_{10}+n_{11}$. Therefore, we have $N \sum e_{ik} a_k - \sum e_{ik} \sum a_k = -(n_{00} + n_{01} + n_{10} + n_{11})n_{10} - (n_{01}-n_{10})(n_{10}+n_{11}) = -n_{00}n_{10} - 2n_{01}n_{10} - n_{01}n_{11} < 0$, and accordingly, $Cov(e^{(i)}, X) < 0$.
\end{proof}

\theoremBinaryC*

\begin{proof}
Similarly, we prove this theorem for a given sample with size $N$. Again, we write the ground truth as $X = \{a_k\}_{k=1}^N$, and write the prediction vector and error vector of tree $i$ as $\widehat{X}^{(i)} = \{p_{ik}\}_{k=1}^N$, $e^{(i)} = \{e_{ik}\}_{k=1}^N$. First, we lay out all possible value combinations of $a_k, p_{ik}, p_{jk}, e_{ik}, e_{jk}$ in the following table:

\begin{table}[htb]
\centering
\begin{tabular}{c|c|c|c|c|c|c}
\hline
$a_k$ & $p_{ik}$ & $p_{jk}$ & Count & $e_{ik}$ & $e_{jk}$ & Abbr. Count Notation  \\
\hline
0 & 0  & 0  & $n_{000} = N \times p_{000}$  & 0  & 0  & $n_{1}$  \\
0 & 1  & 0  & $n_{010} = N \times p_{010}$  & 1  & 0  & $n_{2}$  \\
0 & 0  & 1  & $n_{001} = N \times p_{001}$  & 0  & 1  & $n_{3}$  \\
0 & 1  & 1  & $n_{011} = N \times p_{011}$  & 1  & 1  & $n_{4}$  \\
1 & 0  & 0  & $n_{100} = N \times p_{100}$  & -1 & -1 & $n_{5}$  \\
1 & 1  & 0  & $n_{110} = N \times p_{110}$  & 0  & -1 & $n_{6}$  \\
1 & 0  & 1  & $n_{101} = N \times p_{101}$  & -1 & 0  & $n_{7}$  \\
1 & 1  & 1  & $n_{111} = N \times p_{111}$  & 0  & 0  & $n_{8}$  \\
\hline
\end{tabular}
\end{table}

Next, write $Cov(e^{(i)}, e^{(j)}) = \dfrac{1}{N^2} (N \sum e_{ik} e_{jk} - \sum e_{ik} \sum e_{jk})$. Note that $\sum e_{ik} e_{jk} = n_4+n_5$, $\sum e_{ik} = (n_2+n_4)-(n_5+n_7)$, and $\sum e_{jk} = (n_3+n_4)-(n_5+n_6)$. Then, $N \sum e_{ik} e_{jk} - \sum e_{ik} \sum e_{jk} = (n_1 + \dots + n_8)(n_4+n_5) - [(n_2+n_4)-(n_5+n_7)][(n_3+n_4)-(n_5+n_6)]$. Denote $A = (n_1 + \dots + n_8)(n_4+n_5)$ and $B = [(n_2+n_4)-(n_5+n_7)][(n_3+n_4)-(n_5+n_6)]$, we calculate the two quantities separately as follows.

First, we rewrite 
\begin{equation*}
\begin{split}
A & = (n_1 + \dots + n_8)n_4 + (n_1 + \dots + n_8)n_5 \\
  & = (n_1 + n_8)(n_4+n_5) + (n_2+n_3+n_4)n_4 + (n_5+n_6+n_7)n_4 + (n_2+n_3+n_4)n_5 + (n_5+n_6+n_7)n_5 \\
\end{split}
\end{equation*}

Second, we rewrite
\begin{equation*}
\begin{split}
B & = (n_2n_3 + n_2n_4 + n_3n_4 + n_4^2) + (n_6n_7 + n_5n_6 + n_5n_7 + n_5^2) \\ & \quad \quad - (n_2n_5 + n_4n_5 + n_2n_6 + n_4n_6) - (n_3n_5 + n_3n_7 + n_4n_5 + n_4n_7)\\
  & =  (n_2n_3+n_6n_7 - n_2n_6 - n_3n_7) + (n_2+n_3+n_4)n_4 + (n_5+n_6+n_7)n_5 \\ & \quad \quad - (n_5+n_6+n_7)n_4 - (n_2+n_3+n_4)n_5\\
\end{split}
\end{equation*}

Together, we have $N \sum e_{ik} e_{jk} - \sum e_{ik} \sum e_{jk} = A-B = (n_1 + n_8)(n_4+n_5) + 2(n_5+n_6+n_7)n_4 + 2(n_2+n_3+n_4)n_5 + (n_2n_6 + n_3n_7 - n_2n_3 - n_6n_7) = (n_1 + n_8)(n_4+n_5) + 2(n_5+n_6+n_7)n_4 + 2(n_2+n_3+n_4)n_5 + (n_2-n_7)(n_6-n_3)$. Using the original count notations, the right-hand-side is equivalent to $(n_{000}+n_{111})(n_{011}+n_{100}) + 2(n_{0 \bullet \bullet}-n_{000})n_{100} + 2(n_{1 \bullet \bullet}-n_{111})n_{011} + (n_{010}-n_{101})(n_{110}-n_{001})$. Therefore, we have $Cov(e^{(i)}, e^{(j)}) > 0 \Leftrightarrow \frac{1}{N^2} [(n_{000}+n_{111})(n_{011}+n_{100}) + 2(n_{0 \bullet \bullet}-n_{000})n_{100} + 2(n_{1 \bullet \bullet}-n_{111})n_{011} + (n_{010}-n_{101})(n_{110}-n_{001})] > 0 \Leftrightarrow (p_{000}+p_{111})(p_{011}+p_{100}) + 2(p_{0 \bullet \bullet}-p_{000})p_{100} + 2(p_{1 \bullet \bullet}-p_{111})p_{011} + (p_{010}-p_{101})(p_{110}-p_{001}) > 0$.
\end{proof}

\clearpage

\section{Using ForestIV in Practice} \label{Practice}
In practice, because true coefficients are not known \textit{a priori}, it is useful to have a few guidelines to gauge the effectiveness of ForestIV in a particular finite sample.

\begin{enumerate}
  \item Using the hold-out dataset (e.g., $D_{test}$), researchers can empirically evaluate instrument validity and strength, both before and after the proposed two-step lasso-based selections.
  \item The Hotelling $T^2$ test statistic can also be a useful signal. The $p$-value associated with the Hotelling $T^2$ test comparing $\widehat{\boldsymbol{\beta}}_{label}$ and ForestIV estimates indicates the probability of observing their empirical differences under the null of equality. Researchers can define the threshold of evidence they require before accepting ForestIV estimates, by adjusting the significance level for this test. 
  \item Researchers can also examine whether the asymptotic properties of ForestIV have yet to ``kick-in" by examining empirical convergence in resulting coefficient estimates as the procedure is exposed to more unlabeled data. If a convergence plot indicates that the coefficient estimates have not yet plateaued, this may be a sign that ForestIV estimates have not yet converged, and that more unlabeled data is perhaps needed. 
\end{enumerate}

Finally, to better delineate the use case of ForestIV, we reiterate that if the size of labeled data is large enough that $\widehat{\boldsymbol{\beta}}_{label}$ can be estimated reliably and precisely enough using only the available labeled data, then there is no need to mine variables in the first place. One should simply take $\widehat{\boldsymbol{\beta}}_{label}$ and proceed with inferences and decision making. Accordingly, statistical power analyses are likely to be useful when determining whether ``big data" and machine learning methods are needed for a particular inference problem \citep{ellis2010essential}.

\clearpage

%%%%%%%%%Bibliography%%%%%%%%%%
\bibliographystyle{apalike}%%%%
\bibliography{SSRN.bib}%%%
%%%%%%%%%%%%%%%%%%%%%%%%%%%%%%%

%%%%%%%%%%%%%%
\end{document}